\documentclass[twoside,11pt]{article}
\usepackage{jair, theapa, rawfonts}

\usepackage{amssymb,latexsym}
\usepackage{graphicx,amsmath,amsfonts,amsthm}
\usepackage{comment}
\usepackage{tikz}
\usepackage{pgfplots}
\usepgfplotslibrary{fillbetween}
\usepackage{url}
\usepackage{mathrsfs}
\usepackage{paralist}
\usepackage[]{units}
\usepackage{booktabs}
\usepackage{setspace}
\usepackage{xspace}
\usepackage{subcaption}
\usepackage{tabularx}
\allowdisplaybreaks
\usepackage[linesnumbered,lined,noend]{algorithm2e}
\usepackage{kbordermatrix}

\newtheorem{theorem}{Theorem}

\newtheorem{corollary}{Corollary}
\newtheorem{lemma}{Lemma}

\newtheorem{definition}{Definition}
\newtheorem{example}{Example}
\newtheorem{proposition}{Proposition}

\makeatletter
\newtheorem*{rep@theorem}{\rep@title}
\newcommand{\newreptheorem}[2]{%
\newenvironment{rep#1}[1]{%
 \def\rep@title{#2 \ref{##1}}%
 \begin{rep@theorem}}%
 {\end{rep@theorem}}}
\makeatother
\newreptheorem{theorem}{Theorem}
\newreptheorem{proposition}{Proposition}
\newreptheorem{lemma}{Lemma}
\newreptheorem{corollary}{Corollary}
\newtheorem{construction}{Construction}

\newcommand{\cprob}[3]{
\begin{center}
{\small 
\begin{tabularx}{0.98\columnwidth}{ll}
\toprule
\multicolumn{2}{c}{\textsc{#1}} \\
\midrule
{\bf Instance:}   & \parbox[t]{0.8\columnwidth}{#2\vspace*{1mm}}  \\
{\bf Question:}& \parbox[t]{0.8\columnwidth}{#3\vspace*{.5mm}} \\ 
\bottomrule
\end{tabularx}
}
\end{center}
}
\newcommand{\calO}{{\mathcal{O}}}
\newcommand {\bigO} {{\mathcal O}}
\newcommand{\NP}{NP}

\newcommand{\true}{\textsl{true}\xspace}
\newcommand{\false}{\textsl{false}\xspace}
\newcommand{\nsp}{\textnormal{\texttt{not\_single\_peaked}}\xspace}
\newcommand{\incompatible}{\textnormal{\texttt{incompatible}}\xspace}
\newcommand{\card}[1]{\ensuremath{|#1|}}

\renewcommand{\leq}{\le}

\newcommand{\citep}[1]{\cite{#1}}
\newcommand{\citet}[1]{\citeA{#1}}

\title{Incomplete Preferences in Single-Peaked Electorates}
\author{\name Zack Fitzsimmons  \email zfitzsim@holycross.edu\\
   \addr College of the Holy Cross \\
   Department of Mathematics and Computer Science\\
   One College Street, Worcester, MA 01610, USA
  \AND
\name Martin Lackner \email lackner@dbai.tuwien.ac.at\\
  \addr TU Wien\\
  Databases and Artificial Intelligence Group\\
  Favoritenstra\ss e 9-11, 1040 Vienna, Austria
}

\jairheading{67}{2020}{797-833}{07/2019}{04/2020}
\ShortHeadings{Incomplete Preferences in Single-Peaked Electorates}
{Fitzsimmons \& Lackner}
\firstpageno{797}

\date{}
\begin{document}

\maketitle

\begin{abstract}
Incomplete preferences are likely to arise in real-world preference aggregation scenarios.
This paper deals with determining whether an incomplete preference profile is single-peaked.
This is valuable information since many intractable voting problems become tractable given single-peaked preferences.
We prove that the problem of recognizing single-peakedness is \NP-complete for incomplete profiles consisting of partial orders.
Despite this intractability result, we find several polynomial-time algorithms for reasonably restricted settings.
In particular, we give polynomial-time recognition algorithms for weak orders, which can be viewed as preferences with indifference.
\end{abstract}


\section{Introduction}

Both human and automated decision making often have to rely on incomplete information.
The same issue arises in collective decision making---voting---in multi-agent systems.
\citet{kon-lan:c:incomplete-prefs} distinguish two main sources of incompleteness:
The first one is \emph{intrinsic} incompleteness where a voter is unable
or unwilling to give complete information. %
The second one is \emph{epistemic} incompleteness where voters are able to provide complete information but at the time of decision making the required information is not fully available.
Clearly, also a combination of these two scenarios may arise.

Incomplete qualitative preferences are often modeled as partial orders \cite<see, e.g., the survey by>{brandt2015handbook-chapter10}.
For example, if voters provide some pairwise comparisons between alternatives, then partial orders arise naturally as the transitive closure of these comparisons.
In this sense, partial orders can be viewed as the most general model for (qualitative) incomplete preferences while maintaining transitivity as a rationality assumption.
A further standard assumption in this setting~\citep{brandt2015handbook-chapter10} is that voters have a true, \emph{complete} preference order, modeled as a total order.
This total order is not available (due to intrinsic or epistemic incompleteness), but the known partial order is compatible with the underlying total order, i.e., it can be extended to this total order.

Real-world preferences---whether they are incomplete or not---tend to possess an underlying structure, as they are formed by an implicit or explicit rationale.
In this paper, we consider a particularly advantageous structural property: single-peaked preferences~\citep{bla:j:rationale-of-group-decision-making}.
The single-peaked domain is one of the most commonly studied preference restrictions in social choice theory~\citep{gaertner-domainrestrictions}.
Single-peaked preferences arise if preferences are formed based on an underlying ordering of alternatives, a so-called axis.
It is assumed that voters have a most-preferred alternative, their peak, and alternatives farther away on this axis are increasingly less preferred.

Single-peaked preference profiles enjoy many positive properties \cite<see, e.g., surveys by>{elkind2016preferencerestrictions,ELP-trends}, both for algorithmic purposes and from the axiomatic point of view from classical social choice.
For example, the problem of determining if a candidate is a winner for Dodgson, Kemeny, and Young elections is computationally hard to compute~\citep{bar-tov-tri:j:who-won,HemaspaandraHR97,HemaspaandraSV05,rot-spa-vog:j:young}, but efficient, polynomial-time algorithms exist when preferences are single-peaked~\citep{bra-bri-hem-hem:j:sp2}.
Another example of a desirable property is that single-peaked profiles are guaranteed to possess a transitive majority relation and thus a (weak) Condorcet is guaranteed to exist \citep{bla:j:rationale-of-group-decision-making,bla:b:polsci:committees-elections}.
Furthermore, for single-peaked profiles there exist nondictatorial strategyproof voting rules~\citep{moulin1984generalized,mou:b:axioms} (e.g., selecting a Condorcet winner); such rules do not exist for arbitrary preference profiles by the Gibbard--Satterthwaite Theorem \citep{gib:j:polsci:manipulation,sat:j:polsci:manipulation}.

If only incomplete preference information is available, it is a challenge to determine whether voters' preferences can be explained by a single-peaked axis.
This is the main goal of this paper: to investigate the algorithmic problem of recognizing whether an incomplete preference profile may be single-peaked, i.e., whether the voters' underlying complete preference orders are single-peaked.
The motivation to obtain this information is twofold:
First, to know that incomplete preferences may be explainable by a single-peaked axis provides  fundamental understanding about the nature of the given preference data set, since
the preferences of all voters can be explained by optimal points (peaks) on such an axis.
Second, the single-peaked combinatorial structure is highly advantageous for algorithmic purposes, also in the case of incomplete preferences.
This claim has been corroborated by recent work, which we discuss below in Section~\ref{sec:related}.
In contrast to these two positive aspects, the axiomatic benefits of the single-peaked domain do \emph{not} fully extend to incomplete preferences. In particular, a transitive majority relation is no longer guaranteed (as we discuss later in Section~\ref{subsec:trans}). It thus deserves special attention that the algorithmic 
usefulness of incomplete single-peaked preferences is not an immediate consequence of their social choice-theoretic properties.

We now give an overview of the main results in this paper.

\subsection{Results}

This paper deals with the question of how to determine whether an incomplete preference profile (e.g., a preference profile comprised of partial orders) %
is single-peaked.
This requires a definition of single-peakedness that is applicable to more general types of orders than total orders, i.e., partial orders.
To this end, we propose the notion of \emph{possibly single-peakedness}, which is satisfied if---given an incomplete preference profile---there exist a compatible complete profile that is single-peaked.
To be more precise, we ask whether a given profile of partial orders can be extended to total orders such that the resulting complete profile is single-peaked (according to the usual definition).
Based on this definition, we analyze the corresponding algorithmic recognition problem.
In the following, let $n$ denote the size of the preference profile (number of voters) and let $m$ denote the number of candidates.
Our main computational results are as follows:

\begin{itemize}

\item
We prove that determining whether an incomplete preference profile is single-peaked is \NP-complete.
This is in contrast to the case of complete preferences for which single-peakedness can be determined in linear time~\citep{esc-lan-ozt:c:single-peaked-consistency}.
Furthermore, we strengthen this result by showing that \NP-completeness still holds if one voter completely specifies his or her preferences, and also for the case of local weak orders.

\end{itemize}

Apart from these hardness results, we present polynomial-time algorithms
for recognizing single-peakedness in relevant subcases:

\begin{itemize}

\item First, we consider the problem of determining whether a profile of partial orders is single-peaked for a given axis. 
We prove this problem to be polynomial-time solvable even for incomplete profiles consisting of partial orders.

\item Our most general polynomial-time algorithm is based on the consecutive ones problem. By providing a reduction from single-peaked recognition to the consecutive ones problem, we obtain an $\bigO(m^2\cdot n)$ algorithm for profiles of weak orders, i.e., rankings with ties.
We have implemented this algorithm; see the discussion section (Section~\ref{sec:disc}) for further details.

\item We also provide a direct $\bigO(m\cdot n)$ algorithm for single-peaked consistency for
    weak orders, with the requirement that the preference profile contains
    at least one implicitly specified total order.
This algorithm is an improvement over the single-peaked recognition algorithm by~\citet{esc-lan-ozt:c:single-peaked-consistency} since it is applicable to a broader class of preference profiles (weak orders instead of total orders) while maintaining its runtime.

\item
    Furthermore, we present a \textsc{2-SAT}-based algorithm,
which also requires an implicitly specified
total order, but is applicable to local weak orders (a generalization of weak orders).
This more general algorithm requires $\bigO(m^3\cdot n)$ time.

\item For the case of top orders (weak orders with all incomparable candidates ranked last), we provide another direct, combinatorial algorithm with a runtime of $\bigO(m^2\cdot n)$. 
This algorithm does not rely on solving consecutive ones instances and thus
is conceptually simpler, in particular for implementations.

\end{itemize}

These algorithms differ in their worst-case runtime estimates, which range from $\bigO(m\cdot n)$ to $\bigO(m^3\cdot n)$. Note that these runtime estimates differ only with respect to the number of candidates. In particular in scenarios with a large number of candidates \cite<e.g., when aggregating web search results, as studied by>{BetzlerBN14}, the actual runtime of these algorithms can be expected to vary.

\smallskip

Finally, we consider the case that voters are unable to provide total orders because their true preferences include some form of indifference. This gives another perspective on weak orders: viewing them not as incomplete but as rankings with ties.
We compare our definition of possibly single-peaked preferences to other applicable single-peaked concepts, highlighting benefits and drawbacks of our more general approach.

\subsection{Related Work}
\label{sec:related}

As mentioned before, the single-peaked property is highly advantageous for algorithmic purposes.
This has been demonstrated for many social choice problems based on total (linear) order preferences, e.g., in the work of~\citet{bet-sli-uhl:j:mon-cc,bra-bri-hem-hem:j:sp2}, and~\citet{elkind2015owa}.
Recent work has shown that similar benefits also hold for possibly single-peaked preferences, in particular for profiles of weak orders.
\citet{peters2018single} shows that several rules that are NP-hard otherwise can be computed in polynomial time for possibly single-peaked profiles of weak orders. This holds for multi-winner rules such as Chamberlin-Courant, Proportional Approval Voting, and the very general class of OWA-based rules.
Furthermore, Young's rule is also polynomial-time computable for possibly single-peaked profiles of weak orders; this follows from work of \citet{peters2016spoc}, who prove this result for the more general domain of preferences that are single-peaked on a circle.
If we restrict weak orders to two equivalence classes, i.e., we only distinguish between good (approved) and bad (disapproved) candidates, we obtain the class of dichotomous preferences. The notion of possibly single-peakedness directly translates to this setting. \citet{liu2016parameterized} have shown that Minimax Approval Voting, although NP-hard in general~\citep{legrand2007some}, can be solved in polynomial time for possibly single-peaked dichotomous preferences (see the work of \citeR{elkind2015structure} containing a detailed discussion of single-peakedness and dichotomous preferences).
Overall, the notion of possibly single-peaked preferences has been shown to be beneficial for many algorithmic purposes.

Another interesting effect of single-peaked preferences is that the complexity of manipulative actions (such as manipulation and control) often decreases for single-peaked preferences~\citep{fal-hem-hem-rot:j:shield,hemaspaandra2016complexity}.
When we instead consider more general types of preference orders than total orders we get different behavior
depending on the model of single-peakedness used. In general, the model of possibly single-peakedness does not yield such a reduction 
in complexity compared to the general case for different manipulative actions~\citep{fit-hem:c:voting-with-ties,FitzsimmonsH16,men-lar:j:top-truncated-sp}. This is in contrast to the algorithmic benefits mentioned in the last paragraph.
However, here a decrease in complexity would not necessarily be desirable.
In particular, if one seeks to reduce insincere behavior by choosing voting rules that are computationally hard to manipulate, then a reduction in computational complexity removes this kind of protection~\citep{fal-hem-hem-rot:j:shield}.
We also mention that~\citet{FitzsimmonsH16} consider the complexity of manipulative actions of the models of single-peakedness
discussion in Section~\ref{sec:other} as well as another related restriction, single-peaked preferences with outside
options~\citep{can:j:single-peaked-outside-option}.

Recognition algorithms for restricted preference domains are also indispensable for characterization results. The single-peaked domain~\citep{bal-hae:j:characterization-single-peaked}, the single-crossing domain~\citep{bre-che-woe:j:single-crossing}, and the group-separable domain~\citep{bal-hae:j:characterization-single-peaked} have all been characterized by a set of \emph{forbidden subprofiles}.
These characterizations are obtained by analyzing under which conditions recognition algorithms certify that a profile is \emph{not} contained in a preference domain.
The Unguided Algorithm, as presented in this paper, was used for characterizing preferences that are single-peaked on a circle~\citep{peters2016spoc}.
This was possible by relating the problem of recognizing preferences that are single-peaked on a circle to the problem of recognizing possibly single-peaked preferences.
Such characterizations do not always exist, e.g., the Euclidean domain cannot be characterized by a finite set of forbidden subprofiles~\citep{chen2017one}.

Our definition of possibly single-peaked preferences resembles that of a possible winner \citep{kon-lan:c:incomplete-prefs}: A candidate is a possible winner if---given an incomplete preference profile---there exists a compatible complete profile in which this candidate is a winner.
The possible winner problem is a central, well-studied problem in computational social choice~\citep{wal:c:uncertainty-in-preference-elicitation-aggregation,betzler2009multivariate,betzler2010towards,pini2011incompleteness,xia2011determining,Baumeister2012186,dery2014reaching}, see also the survey by \citet{brandt2015handbook-chapter10}.

The single-peaked domain is far from the only well-studied preference domain; many other domains have been considered from a computational viewpoint.
We refer to the survey of \citet{ELP-trends} for an overview.
Here, we briefly mention papers that build upon or are directly related to our work.
Another major preference restriction is the single-crossing domain, characterized by an ordering of voters.
\citet{elkind2015incompletecrossing} analyzed the recognition of incomplete single-crossing preferences. This work studies possibly single-crossing profiles (a notion analogously defined to our definition of possibly single-peaked profiles), but found that approaches like consecutive ones and \textsc{2-SAT} reductions---which we use here---yield weaker results in the single-crossing domain. In particular, fixing the underlying order appears to have a much weaker effect on the recognition problem.
Another notion is top-monotonicity \citep{bar-mor:j:top-monotonicity}, for which questions similar to ours have been studied.
The results for top-monotonicity resemble those for single-peaked preferences, although new techniques are required: The recognition problem for partial orders is NP-hard \citep{aziz2014testing} and solvable in polynomial time for weak orders~\citep{magiera2019recognizing}.
Furthermore, \citet{peters2016recognising} highlights the inherent complexity of detecting more-dimensional Euclidean preferences and discusses possible definitions for weak orders.

More broadly, recommender systems routinely have to deal with incomplete preferences and therefore this issue has received significant attention in the field. We refer to the following works for an overview \citep{schafer2007collaborative,de2010learning,Masthoff2015GroupRecommenderSystems}.

\subsection{Outline of the Paper}

In Section~\ref{sec:prel}, we provide basic definitions from social choice theory. This is followed by Sec\-tion~\ref{sec:sp}, in which we introduce our notion of possibly single-peaked preferences and study their basic properties.
In Section~\ref{sec:hardness}, we prove the computational hardness of recognizing possibly single-peaked preferences.
All algorithmic results are then presented in Section~\ref{sec:algo}; this section contains the most important contributions of this paper.
In Section~\ref{sec:other}, we discuss indifference, weak orders, and other models of single-peakedness.
A final discussion of our results, a note on an implementation of the consecutive ones approach, and suggestions for future research can be found in Section~\ref{sec:disc}. Some proof details (correctness proofs of some algorithms) are delegated to an appendix, Sections~\ref{sec:app:corr-guided}, and~\ref{sec:app:unguided}.

\section{Preliminaries}
\label{sec:prel}

In this paper, preferences are represented by different types of orders (see Figure~\ref{fig:zoo} for examples).
The most general type are partial orders:
a \emph{(strict) partial order} $P$ on a set $X$ is a binary relation on $X$ that
is transitive ($xPy$ and $yPz$ implies $xPz$ for all $x,y,z\in X$) and asymmetric (if $xPy$ then not $yPx$ for all $x,y\in X$).
We say that \emph{$y$ is ranked above $x$} if $xPy$ holds.
If for two elements $x,y\in X$ neither $xPy$ nor $yPx$ holds, we say that these two elements are \emph{incomparable}; we write $x\sim y$.
An element $x\in X$ is \emph{minimal} if there is no element $y\in X$ with $yPx$; maximal elements are defined analogously.

A \emph{(strict) weak order} is a partial order where the incomparability relation is transitive.
Weak orders have the same expressiveness as total preorders\footnote{Total preorders are binary relations that are transitive and complete ($xPy$ or $yPx$ for all $x,y\in X$).
For our algorithmic purposes, it is not relevant whether to use strict weak-orders or total preorders.
Our main reason to use strict weak orders as the standard definition is that total preorders are not necessarily partial orders (they are not asymmetric)
and hence do not fit in our hierarchy of orders.} (also known as nonstrict weak orders or preference orders).
Since a weak order can be considered a ranking with ties, we refer to $\sim$ also as ``indifference'' in this context
(see Section~\ref{sec:other}, where we discuss
several other models of single-peakedness for preferences with indifference).
Weak orders are also referred to as bucket orders (e.g. in \citeR{fagin2006comparing}); elements that tie are in the same ``bucket''.

A \emph{top order} is a weak order where incomparability appears only among minimal elements, i.e., if $a\sim b$ then both $a$ and $b$ are minimal elements in a top order.
The \emph{ranked} (i.e., nonminimal) elements of a top order $T$
are those that are not incomparable to any other candidate.
We would like to remark that top orders are also known as \emph{top lists} \citep{dwo-kum-nao-siv:c:rank-aggregation,fagin2003comparing} and as top-truncated votes \citep{Baumeister2012}.

A partial order with no incomparable elements is called \emph{total order}, or, equivalently, a total order is a complete ($xPy$ or $yPx$ for all $x,y\in X$) partial order.
Any partial order $P$ can be \emph{extended} to some total order $T$ such that $aPb$ implies $aTb$; $T$ is then a (not necessarily unique) \emph{extension} of $P$.

Finally, we define a \emph{local weak order} $P$ on a set $X$ to be a partial order on $X$ with the following property:
there exist sets $X_1,X_2$ with $X_1\cup X_2=X$ such that the elements in $X_1$ are incomparable to all other elements in $X$ and the profile $P$ restricted to $X_2$ is a weak order.
Intuitively, a local weak order is a weak order together with some isolated elements for which absolutely no information is available.

Throughout the paper, total orders are denoted by $\left\langle c_1 > c_2 > \ldots > c_k \right\rangle$; the brackets allow us to unambiguously  denote total orders consisting of one or even zero elements, i.e, we use $\left\langle \right\rangle$ to denote the empty order relation.
For top orders, we write $\left\langle c_1 > c_2 > \ldots > c_k > \bullet \right\rangle$ to denote a top order where $c_1, \ldots, c_k$ are ranked as stated and all other elements (usually all remaining candidates in $C$) are ranked last, i.e., are minimal elements.
We sometimes use set operators ($\cup,\cap,\setminus$)
on top orders with the intended meaning that we apply these operators to the corresponding sets of ranked candidates.
We describe weak orders with a notation similar to top orders and use $\sim$ to signify incomparability, e.g., $\left\langle c_1 > c_2 \sim c_3 > c_4 \right\rangle$.

\begin{figure}
\centering
{\renewcommand{\tabcolsep}{1.4em}
{
\begin{tabular}{ccccc}
	Total & Top & Weak & Local weak & Partial \\
	order & order & order & order & order \\
	\begin{tikzpicture}
	\tikzstyle{every node}=[circle, draw=black, fill=black,inner sep= 0.07cm];
	\tikzstyle{every path}=[thick]
	\node (5) at (0,5*0.7) {};
	\node (4) at (0,4*0.7) {};
	\node (3) at (0,3*0.7) {};
	\node (2) at (0,2*0.7) {};
	\node (1) at (0,1*0.7) {};
	\node (0) at (0,0*0.7) {};
	\draw (0) -- (1) -- (2)-- (3) -- (4) -- (5);
	\end{tikzpicture}
	& 
	\begin{tikzpicture}
	\tikzstyle{every node}=[circle, draw=black, fill=black,inner sep= 0.07cm];
	\tikzstyle{every path}=[thick]
	\node (5) at (0,5*0.7) {};
	\node (4) at (0,4*0.7) {};
	\node (3) at (0,3*0.7) {};
	\node (2) at (0,2*0.7) {};
	\node (1) at (0,1*0.7) {};
	\node (0a) at (0.4,0) {};
	\node (0b) at (0,0) {};
	\node (0c) at (-0.4,0) {};
	\draw (0a) -- (1) -- (2)-- (3) -- (4) -- (5);
	\draw (1) --(0b);
	\draw (1) --(0c);
	\end{tikzpicture}
	& 
	\begin{tikzpicture}
	\tikzstyle{every node}=[circle, draw=black, fill=black,inner sep= 0.07cm];
	\tikzstyle{every path}=[thick]
	\node (5a) at (-0.4,5*0.7) {};
	\node (5b) at (0.4,5*0.7) {};
	\node (4) at (0,4*0.7) {};
	\node (3a) at (-0.4,3*0.7) {};
	\node (3b) at (0.4,3*0.7) {};
	\node (3c) at (0,3*0.7) {};
	\node (2) at (0,2*0.7) {};
	\node (1) at (0,1*0.7) {};
	\node (0a) at (-0.4,0) {};
	\node (0b) at (0,0) {};
	\node (0c) at (0.4,0) {};
	\draw (0a) -- (1) -- (2)-- (3a) -- (4) -- (5a);
	\draw (1) --(0b);
	\draw (1) --(0c);
	\draw (5b) --(4)--(3b)--(2);
	\draw (4)--(3c)--(2);
	\end{tikzpicture} 
	& 
	\begin{tikzpicture}
	\tikzstyle{every node}=[circle, draw=black, fill=black,inner sep= 0.07cm];
	\tikzstyle{every path}=[thick]
	\node (5a) at (-0.4,5*0.7) {};
	\node (5b) at (0.4,5*0.7) {};
	\node (4) at (0,4*0.7) {};
	\node (4r) at (0.3,1.5*0.7) {};
	\node (4s) at (0.6,1.5*0.7) {};
	\node (3a) at (-0.4,3*0.7) {};
	\node (3b) at (0.4,3*0.7) {};
	\node (3c) at (0,3*0.7) {};
	\node (2) at (0,2*0.7) {};
	\node (1) at (0,1*0.7) {};
	\node (0a) at (-0.4,0) {};
	\node (0b) at (0,0) {};
	\node (0c) at (0.4,0) {};
	\draw (0a) -- (1) -- (2)-- (3a) -- (4) -- (5a);
	\draw (1) --(0b);
	\draw (1) --(0c);
	\draw (5b) --(4)--(3b)--(2);
	\draw (4)--(3c)--(2);
	\end{tikzpicture} 
	& 
	\begin{tikzpicture}
	\tikzstyle{every node}=[circle, draw=black, fill=black,inner sep= 0.07cm];
	\node (5a) at (-0.4,5*0.7) {};
	\node (5b) at (0.4,5*0.7) {};
	\node (4) at (0,4*0.7) {};
	\node (4r) at (0.4,2.3*0.7) {};
	\node (4s) at (0.4,1.3*0.7) {};
	\node (3a) at (-0.4,3*0.7) {};
	\node (3b) at (0.4,3*0.7) {};
	\node (3c) at (0,3*0.7) {};
	\node (2) at (0,2*0.7) {};
	\node (1) at (0,1*0.7) {};
	\node (0a) at (-0.4,0) {};
	\node (0b) at (0,0) {};
	\node (0c) at (0.4,0) {};
	\draw (0a) -- (1) -- (2)-- (3a) -- (4) -- (5a);
	\draw (1) --(0b);
	\draw (1) --(0c);
	\draw (5b) --(4)--(3b);
	\draw (4)--(3c)--(2);
	\draw (4r)--(4s);
	\end{tikzpicture} 
\end{tabular}
}
}
\caption{The order zoo: examples of different types of orders that are used to describe preferences.}
\label{fig:zoo}
\end{figure}
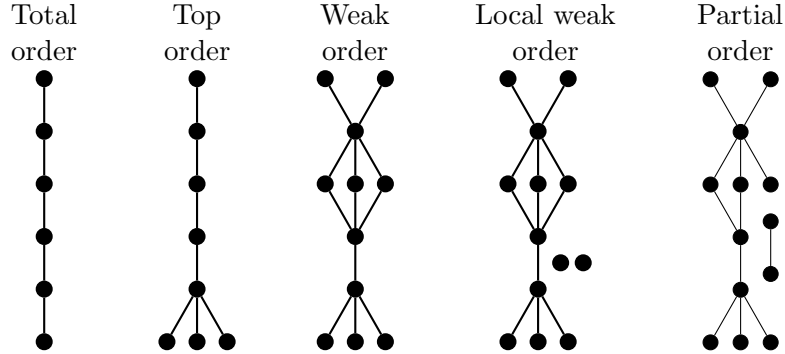

We would now like to address the usefulness of these types of orders for expressing preferences.
Total orders allow the voter to fully specify a strict ranking of options.
Given a large set of options, this might be unfeasible.
Partial orders, on the other hand, allow the voter to specify the relative strict order of any pair of options subject to transitivity, which can be viewed as a rationality constraint.
Thus they can be seen as a very general formalism for representing incomplete preferences.
They are compatible with total orders in the sense that partial orders can always be extended to total orders.
Weak orders are less general than partial orders but arise in many natural scenarios.
For example every real-valued utility function implies a weak order (candidates with the same utility tie, i.e., are incomparable in our sense).
Local weak orders correspond to partial real-valued utility functions and thus arise in scenarios where voters do not have knowledge about all candidates.
If the elicitation of preferences is costly, one might ask only for the most important (top ranked) options of each voter;
in such a case we obtain top orders.
Top orders also are the natural type of order for specifying preferences in some scoring rules.
We will further comment on scoring rules and top orders at the end of Section~\ref{sec:guided}.

Throughout this paper we use $C$ to denote the set of candidates or options.
Votes (or preference orders)
are considered to be either partial, local weak, weak, top, or total orders.
For a vote $V_i$, we use $x\succ_i y$ to denote that $x$ is ranked above $y$.
If there is only one vote under consideration, usually denoted by $V$, we omit the index and write $x\succ y$.
If $V_i$ is a weak or top order, we write $x\sim_i y$ if voter $i$ is indifferent between $x$ and $y$, $x\succsim_i y$ if voter $i$ is either indifferent or ranks $x$ above $y$.

A tuple $(V_1,\ldots,V_n)$ of votes is called a \emph{(preference) profile of $\{$partial orders, local weak orders, weak orders, top orders, total orders$\}$}, depending on the type of orders.
Given a vote $V$ and a set of candidates $C'\subseteq C$, we define $V[C']$ to be the order $V$ restricted to elements in $C'$.
Analogously, given a preference profile $\mathcal{P}=(V_1,\ldots,V_n)$, we define $\mathcal{P}[C']$ to be $(V_1[C'],\ldots,V_n[C'])$.
We denote the number of candidates with $m$ and the number of votes with $n$.

\section{Single-Peaked Profiles}
\label{sec:sp}
We start by giving a definition for single-peaked profiles of total orders and then extend this definition to partial orders.
A central concept is that of an \emph{axis}, which is a total order on $C$.
Let $A$ be an axis.
Throughout this paper we write $x\rhd y$ instead of $xAy$.
(Note that we use $\succ$ for votes and $\rhd$ for axes.)
We formally define single-peakedness for total orders below using triples of candidates, so-called {\em valleys}. Valleys
will be important for many of the results in our paper.
\begin{definition} %
Let $V$ be a partial order on $C$.
A vote $V$ \emph{contains a v-valley with respect to an axis $A$} if there exist $c_1,c_2,c_3\in C$ such that $c_1\rhd c_2\rhd c_3$, $c_1\succ c_2$ and $c_3\succ c_2$.
\end{definition}
The definition of v-valleys suffices to define single-peakedness for profiles of total orders.
\begin{definition}
A profile $\mathcal{P}$ of total orders is single-peaked with respect to $A$ if no vote $V\in\mathcal{P}$ contains a v-valley with respect to $A$ (and thus every vote has only a single ``peak''). A profile of total orders is single-peaked consistent (or simply, single-peaked) if there exists some axis $A$ such that $\mathcal{P}$ is single-peaked with respect to $A$.
\end{definition}
Note that this
is similar to how single-peakedness is defined in~\citep{bra-bri-hem-hem:j:sp2,ELP-trends}, and easily seen to be equivalent.
We now want to extend this definition to profiles of partial orders.
A natural way is to consider extensions of partial orders to total orders.
\begin{definition}
Let $\mathcal{P}=(V_1,\ldots,V_n)$ be a profile of partial orders.
The profile $\mathcal{P}$ is \emph{possibly single-peaked} with respect to an axis $A$ if for every $k\in \{1,\ldots,n\}$, $V_k$ can be extended to a total order $V_k'$ such that the profile of total orders $\mathcal{P'}=(V_1',\ldots,V_n')$ is single-peaked with respect to~$A$.
A profile of partial orders is \emph{possibly single-peaked consistent} if it is possibly single-peaked  with respect to some axis.
\end{definition}

While it is also conceivable to require that every extension is single-peaked (``necessarily'' instead of ``possibly''), this would yield an extremely restrictive definition, basically requiring that the profile of partial orders already contains ``almost'' total orders.
We consider this definition at the end of Section~\ref{sec:other}.
Apart from that, we focus our attention on the more general definition of possibly single-peaked profiles.

We now define an equivalent definition for possibly single-peakedness based on valleys. For this, we first need
the following definition of u-valleys.
\begin{definition} %
\label{def:valley}
Let $V$ be a partial order on $C$.
The vote $V$ \emph{contains a u-valley with respect to $A$} if there exist distinct elements $a,b,c,d\in C$ with $a\rhd  b \rhd  d$ and $a\succ b$ as well as $a \rhd  c \rhd  d$ and $d \succ c$.
\end{definition}
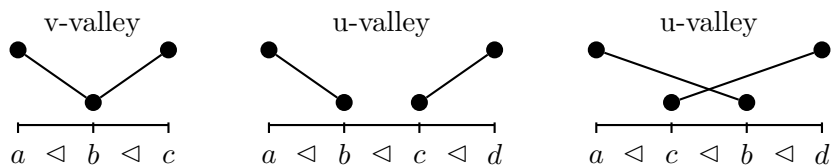
\begin{figure}
\centering
{
\renewcommand{\tabcolsep}{0.8em}
\begin{tabular}{cccc}
	v-valley & u-valley & u-valley \\
	\begin{tikzpicture}[x=1cm]
	\tikzstyle{point}=[circle, draw=black, fill=black,inner sep= 0.07cm];
	\tikzstyle{every path}=[thick]
	\node[point] (0) at (0,1*0.7+0.3) {};
	\node[point] (1) at (1,0*0.7+0.3) {};
	\node[point] (2) at (2,1*0.7+0.3) {};
	\draw (0) -- (1) -- (2);	
	\draw (0,0) -- coordinate (x axis mid) (2,0);
   	\draw (0,1pt) -- (0,-3pt)
		node[anchor=north] {\phantom{I}$a$\phantom{I}};
   	\draw (1,1pt) -- (1,-3pt)
		node[anchor=north] {\phantom{I}$b$\phantom{I}};	
   	\draw (2,1pt) -- (2,-3pt)
 	node[anchor=north] {\phantom{I}$c$\phantom{I}};
 	\foreach \x/\xtext in {0, 1}
		\node[below] at (\x+0.5,0) {$\strut\lhd$};  
	\end{tikzpicture}
	& 
	\begin{tikzpicture}[x=1cm]
	\tikzstyle{point}=[circle, draw=black, fill=black,inner sep= 0.07cm];
	\tikzstyle{every path}=[thick]
	\node[point] (0) at (0,1*0.7+0.3) {};
	\node[point] (1) at (1,0*0.7+0.3) {};
	\node[point] (3) at (2,0*0.7+0.3) {};
	\node[point] (2) at (3,1*0.7+0.3) {};
	\draw (0) -- (1);
	\draw (3) -- (2);	
	\draw (0,0) -- coordinate (x axis mid) (3,0);
   	\draw (0,1pt) -- (0,-3pt)
		node[anchor=north] {\phantom{I}$a$\phantom{I}};
   	\draw (1,1pt) -- (1,-3pt)
		node[anchor=north] {\phantom{I}$b$\phantom{I}};	
   	\draw (2,1pt) -- (2,-3pt)
 	node[anchor=north] {\phantom{I}$c$\phantom{I}};
   	\draw (3,1pt) -- (3,-3pt)
		node[anchor=north] {\phantom{I}$d$\phantom{I}};	
 	\foreach \x/\xtext in {0, 1,2}
		\node[below] at (\x+0.5,0) {$\strut\lhd$};  		
	\end{tikzpicture}
	& 
	\begin{tikzpicture}[x=1cm]
	\tikzstyle{point}=[circle, draw=black, fill=black,inner sep= 0.07cm];
	\tikzstyle{every path}=[thick]
	\node[point] (0) at (0,1*0.7+0.3) {};
	\node[point] (1) at (2,0*0.7+0.3) {};
	\node[point] (3) at (1,0*0.7+0.3) {};
	\node[point] (2) at (3,1*0.7+0.3) {};
	\draw (0) -- (1);
	\draw (3) -- (2);	
	\draw (0,0) -- coordinate (x axis mid) (3,0);
   	\draw (0,1pt) -- (0,-3pt)
		node[anchor=north] {\phantom{I}$a$\phantom{I}};
   	\draw (1,1pt) -- (1,-3pt)
		node[anchor=north] {\phantom{I}$c$\phantom{I}};	
   	\draw (2,1pt) -- (2,-3pt)
 	node[anchor=north] {\phantom{I}$b$\phantom{I}};
   	\draw (3,1pt) -- (3,-3pt)
		node[anchor=north] {\phantom{I}$d$\phantom{I}};	
 	\foreach \x/\xtext in {0, 1,2}
		\node[below] at (\x+0.5,0) {$\strut\lhd$};  				
	\end{tikzpicture}
\end{tabular}%
}%
\caption{Examples of v-valleys and u-valleys}\label{fig:valleys}
\end{figure}
In Figure~\ref{fig:valleys} a graphical comparison of v- and u-valleys is shown.
In the following, we see that these two types of valleys characterize possibly single-peaked incomplete profiles.
\begin{lemma}
\label{lem:singlepeaked-incomplete}
Let $\mathcal{P}=(V_1,\ldots,V_n)$ be a profile of partial orders and $A$ an axis.
The following two statements are equivalent.
\begin{enumerate}
\item[(i)] The profile $\mathcal{P}$ is possibly single-peaked with respect to $A$.
\item[(ii)] Every vote $V\in\mathcal{P}$ contains neither a u-valley nor a v-valley with respect to $A$.
\end{enumerate}
\end{lemma}

\begin{proof}
To see that (i) implies (ii), note that if some $V_k$ contained a u-valley in the sense of Definition~\ref{def:valley} then every extension $V_k'$ would contain a v-valley.
More concretely, if $a,b,c,d\in C$ form a u-valley in $V_k$ then any extension of $V_k$ either contains a v-valley with respect to $A$ on the candidates $a,b,c$ or on $b,c,d$.
Furthermore, if $V_k$ contained a v-valley then so would every extension.

For the other direction, we show that a vote $V$ not containing a valley can be extended to a total order that is single-peaked with respect to $A$.
We recursively define its extension $V'$ starting with its last ranked candidate.
Let $V'(1)$ denote the last ranked candidate of $V'$, $V'(2)$ the second-to-last, etc.
For the definition we require two functions:
\begin{itemize}
\item $\textit{left}_A(X)$ is the smallest (leftmost) candidate in $X$ with respect to the axis $A$.
\item $\textit{right}_A(X)$ is the largest (rightmost) candidate in $X$ with respect to the axis $A$.
\end{itemize}
We now define for every $i\in \{1,\ldots, m\}$,
\begin{align*}
X_i &= C\setminus\{V'(1),\ldots,V'(i-1)\}\text{ and }\\
V'(i) &= \begin{cases} \textit{right}_A\left(X_i\right) & \mbox{if } \textit{right}_A\left(X_i\right)\text{ is a minimal element in } V[X_i],\\
\textit{left}_A\left(X_i\right) &\mbox{otherwise}. \end{cases}.
\end{align*}
This definition immediately yields that $V'$ is single-peaked with respect to $A$:
By always choosing one of the two outermost candidates on $A$ (that have not yet been chosen) for the next higher ranked candidate, valleys cannot arise.

It remains to show that $V'$ is indeed an extension of $V$, i.e., we have to show that for every pair of candidates $a,b\in C$, $a\succ b$ implies $a\succ' b$.
Towards a contradiction assume that $a\succ b$ and $b\succ' a$.
Let $i\in\{1,\ldots,m\}$ such that $V'(i)=a$.
We have to consider two cases: $a=\textit{left}_A(X_i)$ and $a=\textit{right}_A(X_i)$.

Let $a=\textit{left}_A(X_i)$ and $d=\textit{right}_A(X_i)$.
Since $V'(i)=a$, we know that there has to exist a $c\in X_i$ with $d\succ c$.
Observe that $a\rhd c\rhd d$ has to hold.
Furthermore, either $a\rhd b\rhd d$ or $b=d$ holds.
If $a\rhd b\rhd d$ holds then $a,b,c,d$ for a u-valley.
If $b=d$ then $a\rhd c\rhd b$, $a\succ b$ and $b\succ c$ holds: a v-valley.
Both cases contradict our assumption that $V$ does not contain valleys with respect to $A$.

Now, let $a=\textit{right}_A(X_i)$.
This immediately yields a contradiction to the definition of $V'(i)$, since $b\in X_i$, $a\succ b$ and hence $a$ is not a minimal element in $V[X_i]$.
\end{proof}

This lemma immediately yields a polynomial-time algorithm for checking whether an incomplete profile is possibly single-peaked with respect to a given axis:
\begin{proposition}
Verifying whether a profile of partial orders is single-peaked with respect to a given axis can be done in $\mathcal{O}(m^4\cdot n )$ time.
\label{prop:poly-with-given-axis}
\end{proposition}
\begin{proof}
For every quadruple of candidates and every vote, an algorithm has to check whether a u- or v-valley arises (by Lemma~\ref{lem:singlepeaked-incomplete}).
\end{proof}

Let $\mathcal{T}\in \{$partial order, local weak order, weak order, top order, total order$\}$ be a type of order.
The remainder of the paper is dedicated to a more challenging computational problem:
\cprob{\textsc{$\mathcal{T}$ Single-peaked Consistency}}{
A profile $\mathcal{P}$ of type $\mathcal{T}$ and a set of candidates $C$.}{
Is $\mathcal{P}$ possibly single-peaked consistent?}
Note that in contrast to the problem in Proposition~\ref{prop:poly-with-given-axis}, the input of this problem does not include an axis.
The \textsc{Total Order Single-peaked Consistency} problem is known to be solvable in polynomial time as witnessed by several algorithms. %
Historically, the first algorithm to solve this problem was due to \citet{bar-tri:j:stable-matching-from-psychological-model} and is based on the consecutive ones problem, which we will encounter in Section~\ref{sec:consones}.
This approach yields a runtime of $\calO(m^2 \cdot n)$, which was improved by direct, combinatorial algorithms first to $\calO(m \cdot n+m^2)$ by \citet{doignon1994polynomial} and finally to $\calO(m \cdot n)$ by \citet{esc-lan-ozt:c:single-peaked-consistency}.
In the next section, we show that a polynomial-time result is unlikely to exist for the case of partial orders and even local weak orders.

\section{Hardness Results}
\label{sec:hardness}

\begin{theorem}
\textsc{Local Weak Order Single-peaked Consistency} is \textnormal{\textsf{NP}}-complete.\label{thm:SPE-npc}
\end{theorem}

\begin{proof}
We reduce from the \textsf{NP}-complete \textsc{Betweenness} problem~\citep{Opatrny-betweenness}.
A \textsc{Betweenness} instance consists of a finite set $S$ and a set $T$ containing (ordered) triples of distinct elements of $S$.
The decision problem asks whether there is a total order $L$ such that for every triple $(a,b,c)\in T$ we have either $a L b L c$ or $c L b L a$.
Intuitively, a triple $(a,b,c)\in T$ corresponds to the constraint that $b$ has to lie ``in between'' $a$ and $c$ on the total order $L$.

We construct a profile of local weak orders $\mathcal{P}$ with a candidate set $C=S$, i.e., we identify elements in $S$ with candidates in $C$.
The preference profile $\mathcal{P}$ consists of two votes for each triple $(a,b,c)$: the partial order $\{a\succ c, b\succ c\}$ and the partial order $\{b\succ a, c\succ a\}$.
These two votes form a valley on any axis with $c$ between $a$ and $b$ and on any axis with $a$ between $b$ and $c$.
Thus $b$ has to be between $a$ and $c$ on any single-peaked axis.
We are now going to show that $\mathcal{P}$ is possibly single-peaked consistent if and only if the \textsc{Betweenness} instance is a yes-instance.

$``\Rightarrow''$ Assume that there exists an extension of $\mathcal{P}$, $\mathcal{P}^\textit{ext}$, and an axis, $A$, such that $\mathcal{P}^\textit{ext}$ is single-peaked with respect to $A$.
By Lemma~\ref{lem:singlepeaked-incomplete} we know that this implies that no v-valleys exist.
Since for every triple $(a,b,c)\in T$ both the vote $\{a \succ c, b \succ c\}$ and $\{b\succ a, c\succ a\}$ are contained in $\mathcal{P}$, we have that neither $a \rhd  c \rhd  b$, $b \rhd  c \rhd  a$, $b \rhd  a \rhd  c$, nor $c \rhd  a \rhd  b$ can hold.
Consequently it has to hold that either $a \rhd  b \rhd  c$ or $c \rhd  b \rhd  a$ holds and thus $b$ is ``in between'' $a$ and $c$.

$``\Leftarrow''$
Assume that there exists a set $T$ such that all constraints in $T$ are satisfied.
It is easy to verify that $\mathcal{P}$ is possibly single-peaked with respect to $T$.
\end{proof}

\begin{corollary}
\textsc{Partial Order Single-peaked Consistency} is \textnormal{\textsf{NP}}-complete.
\label{cor:POSPC-npc}
\end{corollary}

The proof of Theorem~\ref{thm:SPE-npc} uses preference profiles where the votes contain very little information:
only two pairs of candidates are comparable in each vote.
We know that determining single-peaked consistency is possible in polynomial time if every vote is a total order, i.e., all votes contain complete information.
What happens if one voter provides complete information?
Having a single completely specified vote has been found to be helpful in a related context: It allows for the efficient elicitation of single-peaked preferences using only few comparison queries \citep{con:j:single-peaked} and thus the communication complexity of preference elicitation is reduced.
However, in our case such a voter does not provide enough additional information for a decrease in (computational) complexity.

\begin{theorem}
The \textsc{Partial Order Single-peaked Consistency} problem is \textnormal{\textsf{NP}}-complete, even if the given preference profile contains a total order.
\label{thm:guidedNP}
\end{theorem}

\begin{proof}
We reduce from \textsc{Set Splitting}:
Let $X$ be a finite set.
Given a collection $Z$ of subsets of $X$, is there a partition of $X$ into two subsets $X_1$ and $X_2$ such that no subset of $Z$ is contained entirely in either $X_1$ or $X_2$?
This problem is \NP-complete even if all sets in $Z$ have cardinality three \citep{gar-joh:b:int}.

Let $X=\{c_1,\ldots,c_m\}$.
For the construction,  we identify the elements of $X$ with candidates and add an additional candidate $x$.
For each set $\{c_i,c_j,c_k\}\in Z$ with $i<j<k$ we introduce one vote: $\{c_i\succ c_j, x\succ c_k\}$.
In addition, we add the vote $x\succ c_m \succ \cdots \succ c_1$.
We claim that the resulting preference profile $\mathcal{P}$ is possibly single-peaked consistent if and only if $(X,Z)$ is a \textsc{Set Splitting} yes-instance.

Assume that $\mathcal{P}$ is possibly single-peaked with respect to an axis $A$.
We define $X_1$ to be the candidates on $A$ left of $x$ and $X_2$ those that are right of $x$.
We will show that there is no subset of $Z$ entirely contained in $X_1$ or $X_2$.
Towards a contradiction assume that $\{c_i,c_j,c_k\}\in Z$ with $i<j<k$ are contained in $X_1$.
Then it has to hold that, on $A$, $c_i,c_j,c_k$ are all left of $x$.
Furthermore, from the vote $\left\langle x\succ c_m \succ \cdots \succ c_1\right\rangle$ then it follows that the relative order on $A$ of  $x, c_i,c_j,c_k$ has to be $c_i\rhd c_j\rhd c_k\rhd x$.
However, then the vote $\{c_i\succ c_j, x\succ c_k\}$ contains a u-valley with respect to this order.
Assuming that $\{c_i,c_j,c_k\}\in Z$ are contained in $X_2$ leads to the same contradiction.
Thus, $X_1$ and $X_2$ indeed certify that $(X,Z)$ is a yes-instance.

For the other direction, assume that $(X,Z)$ is a yes-instance and that $X_1$ and $X_2$ are the corresponding partition.
Let an axis $A$ be defined as the elements in $X_1$ with indices in increasing order followed by $x$ followed by the elements in $X_2$ with indices in decreasing order.
We claim that $A$ is an axis for $\mathcal{P}$.
Clearly, the vote $\left\langle x\succ c_m \succ \cdots \succ c_1\right\rangle$ is single-peaked with respect to $A$.
Let us consider a vote $\{c_i\succ c_j, x\succ c_k\}$ with $i<j<k$.
Since $X_1$ and $X_2$ are a valid partition, at least one of $c_i, c_j, c_k$ has to be left of $x$ and another one right.
This rules out that a u-valley is formed and thus all votes are possibly single-peaked with respect to $A$.
\end{proof}

As we will see in the following section, these two hardness results establish the tractability frontier. 

\section{Algorithms}
\label{sec:algo}

In this section, we present several polynomial-time algorithms for recognizing possibly single-peaked profiles.

\subsection{The Consecutive Ones Approach}
\label{sec:consones}

Our first algorithm in this section solves the \textsc{Weak Order Single-peaked Consistency} problem in polynomial time. It uses a reduction to the problem of detecting the \emph{consecutive ones property} in a binary matrix, i.e., a matrix consisting of zeros and ones.
Such a matrix has the consecutive ones property if its columns can be permuted in such a way that in all rows the ones appear consecutively.
The corresponding decision problem is the following:
\cprob{\textsc{Consecutive Ones}}
{A binary matrix $M$.}
{Does $M$ possess the consecutive ones property, i.e., does there exist a permutation of the columns of $M$ such  that in each row all of the ones are consecutive?}

The consecutive ones property was originally defined by \citet{ful-gro:j:contiguous-ones} and shown solvable in polynomial time.
More specifically, they showed that given an $s\times t$ matrix this problem can be solved in $\bigO(s\cdot t^2)$.
\citet{boo-lue:j:consecutive-ones-property}
improved on this result by finding an $\bigO(s\cdot t)$ algorithm\footnote{More precisely, their algorithm has a runtime of $\bigO(s+ t+f)$, where $f$ is the number of ones in $M$. In our case, the matrices obtained from the reduction have $\Theta(s \cdot t)$ one entries and hence this algorithm has a runtime of $\bigO(s\cdot t)$ .} through the
development and use of the novel PQ-tree data structure.
With this data structure it is not only possible to find one valid permutation of columns (if it exists), but to compactly represent all possible column permutations that witness the consecutive ones property.
Subsequent work has improved these results in various ways \citep{meidanis1998consecutive,habib2000lex,hsu2002simple,mcconnell2004certifying,raffinot2011consecutive}.

\citet{bar-tri:j:stable-matching-from-psychological-model} were the first to relate the problem of recognizing single-peaked preferences to the consecutive ones problem.
We slightly modify their approach for profiles of total orders to be applicable to profiles of weak orders, and by that solve the \textsc{Weak Order Single-peaked Consistency} problem.
The reduction works as follows:

\begin{construction}\label{con:mat}
Let $\mathcal P=(V_1,\dots,V_n)$ be a profile of weak orders over candidate set $C$ with $\card{C}=m$.
For each $V_i$ we construct an $m \times m$ binary matrix $X_i$.
We assume that the rows and columns of $X_i$ are indexed by $C$.
For $a,b\in C$, the entry of $X_i$ is defined as:
\[X_i(a,b) = \begin{cases}0 &\text{if }a\succ_i b\\
1 &\text{if }b\succ_i a\text{ or }a\sim_i b\end{cases}.\]
Finally, the matrices $X_1,\dots,X_n$ are 
row-wise concatenated to obtain the
$m n \times m$ matrix~$X_\mathcal{P}$.
\end{construction}
\medskip

\begin{example}\label{ex:extsp} %
Consider the preference profile $\mathcal{P}=(V_1,V_2)$ with 
$V_1=\left\langle a \sim c \succ b \succ e \sim d \succ f\right\rangle$ and
$V_2=\left\langle a \succ b \succ c \succ e \sim d \succ f\right\rangle$.
We construct $X_1$ and $X_2$:

\[X_1 = 
   \kbordermatrix{
     & a & b & c & d & e & f\\
   a & 1 & 0 & 1 & 0 & 0 & 0\\
   b & 1 & 1 & 1 & 0 & 0 & 0\\
   c & 1 & 0 & 1 & 0 & 0 & 0\\
   d & 1 & 1 & 1 & 1 & 1 & 0\\
   e & 1 & 1 & 1 & 1 & 1 & 0\\
   f & 1 & 1 & 1 & 1 & 1 & 1\\
  }
\hspace{0.15\textwidth}
X_2 = 
   \kbordermatrix{ & a & b & c & d & e & f\\
   a & 1 & 0 & 0 & 0 & 0 & 0\\
   b & 1 & 1 & 0 & 0 & 0 & 0\\
   c & 1 & 1 & 1 & 0 & 0 & 0\\
   d & 1 & 1 & 1 & 1 & 1 & 0\\
   e & 1 & 1 & 1 & 1 & 1 & 0\\
   f & 1 & 1 & 1 & 1 & 1 & 1\\
}\]
We then row-wise concatenate $X_1$ and $X_2$ to construct $X_\mathcal{P}$.
\[ X_\mathcal{P} = 
   \kbordermatrix{ & a & b & c & d & e & f\\
    & 1 & 0 & 1 & 0 & 0 & 0\\
    & 1 & 1 & 1 & 0 & 0 & 0\\
    & 1 & 0 & 1 & 0 & 0 & 0\\
    & 1 & 1 & 1 & 1 & 1 & 0\\
    & 1 & 1 & 1 & 1 & 1 & 0\\
    & 1 & 1 & 1 & 1 & 1 & 1\\
    & 1 & 0 & 0 & 0 & 0 & 0\\
    & 1 & 1 & 0 & 0 & 0 & 0\\
    & 1 & 1 & 1 & 0 & 0 & 0\\
    & 1 & 1 & 1 & 1 & 1 & 0\\
    & 1 & 1 & 1 & 1 & 1 & 0\\
    & 1 & 1 & 1 & 1 & 1 & 1\\
}
\hspace{0.25\textwidth}
 X_\mathcal{P}' = 
   \kbordermatrix{ & b & a & c & d & e & f\\
    & 0 & 1 & 1 & 0 & 0 & 0\\
    & 1 & 1 & 1 & 0 & 0 & 0\\
    & 0 & 1 & 1 & 0 & 0 & 0\\
    & 1 & 1 & 1 & 1 & 1 & 0\\
    & 1 & 1 & 1 & 1 & 1 & 0\\
    & 1 & 1 & 1 & 1 & 1 & 1\\
    & 0 & 1 & 0 & 0 & 0 & 0\\
    & 1 & 1 & 0 & 0 & 0 & 0\\
    & 1 & 1 & 1 & 0 & 0 & 0\\
    & 1 & 1 & 1 & 1 & 1 & 0\\
    & 1 & 1 & 1 & 1 & 1 & 0\\
    & 1 & 1 & 1 & 1 & 1 & 1\\
    }\]
Observe that $X_\mathcal{P}$ has the consecutive ones property, as witnessed by $X_\mathcal{P}'$.
The corresponding permutation of columns ($bacdef$) directly corresponds to an ordering of candidates, and indeed
$\mathcal{P}$ is possibly single-peaked with respect to $b \rhd a \rhd c \rhd d \rhd e \rhd f$.
\end{example}

To prove the correctness of Construction~\ref{con:mat}, we first define a valley-based characterization of possibly single-peaked profiles of weak orders, analogous to Lemma~\ref{lem:singlepeaked-incomplete}.
We see that for profiles of weak orders, u-valleys are not relevant.

\begin{lemma}
\label{lem:extv}
Let $\mathcal{P}=(V_1,\ldots,V_n)$ be a profile of weak orders and $A$ an axis.
The following two statements are equivalent.
\begin{enumerate}
\item[(i)] The profile $\mathcal{P}$ is possibly single-peaked with respect to $A$.
\item[(ii)] Every vote $V\in\mathcal{P}$ does not contain a v-valley with respect to $A$.
\end{enumerate}
\end{lemma}
\begin{proof}
That statement~(i) implies (ii) is a special case of Lemma~\ref{lem:singlepeaked-incomplete}.
For the other direction,
every u-valley on $a\rhd b\rhd c\rhd d$ (see~Figure~\ref{fig:valleys}) implies a v-valley either on $a\rhd b\rhd d$ or on $a\rhd c\rhd d$.
\end{proof}

Using Construction~\ref{con:mat} and Lemma~\ref{lem:extv}, we can now show that:

\begin{theorem}\label{thm:weaksp}
The \textsc{Weak Order Single-peaked Consistency} problem can be solved in $\bigO(m^2\cdot n)$ time.%
\end{theorem}

\begin{proof}
Let $\mathcal{P}$ be a preference profile of weak orders. We will show that $V$ is possibly single-peaked if and only if
the matrix $X_\mathcal{P}$, as obtained by Construction~\ref{con:mat}, has the consecutive ones property.
This proof closely follows the argument made by \citet{bar-tri:j:stable-matching-from-psychological-model} for total orders, but requires Lemma~\ref{lem:extv} to relate possible single-peaked weak order to the consecutive ones property.

If $\mathcal{P}$ is possibly single-peaked with respect to an axis $A$
then by Lemma~\ref{lem:extv} we know that
no preference order $V \in \mathcal{P}$ contains a v-valley
with respect to $A$.
When the columns of the matrix $X_\mathcal{P}$ are
permuted with respect to the axis $A$,
no row will contain the sequence
$\cdots 1 \cdots 0 \cdots 1 \cdots$.
since this corresponds to a preference order that strictly decreases
and then strictly increases along the axis $A$ (a v-valley).
Therefore $X$ has the consecutive ones property. 

For the other direction suppose that
$V$ is not possibly single-peaked. Then by Lemma~\ref{lem:extv} we know that for every
possible axis there exists a preference order $V \in \mathcal{P}$ such that $V$ contains a
v-valley with respect to that axis.
So every permutation of the columns of $X_\mathcal{P}$ will
correspond to an axis where some preference order has a v-valley.
As stated in the proof of the other direction, a v-valley
corresponds to a row containing the sequence
$\cdots 1 \cdots 0 \cdots 1 \cdots$, so clearly $X$ does not have the
consecutive ones property.

Since it takes $\bigO(m^2\cdot n)$ time to construct $X_\mathcal{P}$ and $\bigO(m^2\cdot n)$ time to solve the corresponding \textsc{Consecutive Ones} problem,
we can solve the \textsc{Weak Order Single-peaked Consistency} problem in $\bigO(m^2\cdot n)$ time.
\end{proof}

As mentioned earlier, the approach using PQ-trees by \citet{boo-lue:j:consecutive-ones-property} is able to compute all possible permutations that witness the consecutive ones property. Thus, Theorem~\ref{thm:weaksp} actually allows us to compute all possibly single-peaked axes. As these permutations are stored in a compact way, the runtime of $\bigO(m^2\cdot n)$ time holds even for a potentially exponential number of axes.

To sum up, the consecutive ones approach grants us a simple and quick solution for the \textsc{Weak Order Single-peaked Consistency} problem. The downside is that by relying on the full consecutive ones machinery, any practical implementation has to deal with the nontrivial PQ-tree data structure or related concepts \citep{meidanis1998consecutive,habib2000lex,mcconnell2004certifying}.
In the following two sections we present direct, combinatorial algorithms.

\subsection{The Guided Algorithm}
\label{sec:guided}
We now present a second polynomial-time algorithm for profiles of weak orders.
In contrast to consecutive ones approach of Section~\ref{sec:consones}, this algorithm requires an additional condition: The input profile must contain at least one total order, the  \emph{guiding vote}, to guide the placement of candidates on the axis.
Due to this extra information, we gain two major benefits:
First, the algorithm is conceptually simpler than the more complex algorithms for solving the consecutive ones problem.
This has clear benefits for implementing the algorithm, and is also necessary for the use in characterization proofs (see, e.g., the work of \citet{peters2016spoc} as discussed in Section~\ref{sec:related}).
Second, this algorithm is faster. We achieve a runtime of $\bigO(m\cdot n)$ compared to the runtime of $\bigO(m^2\cdot n)$ by the consecutive ones approach.

\begin{theorem}
If the profile contains a total order,  the \textsc{Weak Order Single-peaked Consistency} problem can be solved in $\bigO(m\cdot n)$ time.%
\label{thm:guided-weak}%
\end{theorem}
Theorem~\ref{thm:guided-weak} is based on Algorithm~\ref{alg:untie}, which we refer to as the \emph{Guided Algorithm}.
Without loss of generality, we assume that the guiding vote is $\left\langle c_m \succ c_{m-1} \succ \cdots \succ c_1 \right\rangle$, i.e., we number the candidates based on the guiding vote.
The requirement that this guiding vote is part of the input can be weakened considerably.
It suffices that a guiding vote is implicitly contained in the profile.
The following procedure finds such an implicit guiding vote, if one exists:
Look for some vote with a unique last ranked candidate.
This candidate is ranked last in the guiding vote.
Remove this candidate from the profile and repeat this step to obtain the second-to-last element in the guiding vote, etc.
If at some point no vote has a unique last-ranked candidate, no implicit guiding vote exist.
As in each step there might be more than one last-ranked candidate to choose from, such a guiding vote is not unique.
However, it is straightforward to prove that adding any implicit guiding vote obtained in this way does not change whether the profile is possibly single-peaked.
This procedure can be implemented in $\bigO(m\cdot n)$ time.

\begin{example}\label{ex:guiding}
Let us consider the profile $\mathcal{P}=(V_1,V_2,V_3)$ with
$V_1=\langle a\succ b \succ c \sim d \succ e \rangle$,
$V_2=\langle a\sim b \sim c \succ d \succ e \rangle$, and
$V_3=\langle e\sim d \succ b \succ c \succ a \rangle$.
This profile clearly does not contain a total order. However, a guiding vote
can be found implicitly: 
Candidate $e$ is (uniquely) ranked last in $V_1$, so we remove this candidate.
In this reduced profile, $V_2$ ranks $d$ last.
We again remove this candidate and see that $V_1$ now equals $\langle a\succ b \succ c\rangle$.
Thus, we can use $\langle a\succ b \succ c \succ d \succ e \rangle$ as our
guiding order.
\end{example}

The algorithm has a simple structure:
The lowest ranked candidate in the guiding vote, $c_1$, is placed on the rightmost position of the axis (this choice is arbitrary.)
Starting with the second lowest ranked candidate in the guiding vote, $c_2$, the candidates are successively placed on the axis---either at the leftmost or rightmost still-available position.
The lists $A_L$ and $A_R$ correspond to the left-hand and right-hand side of the axis under construction.
For each candidate, we test whether it can be placed on the right-hand side or left-hand side without creating a v-valley; u-valleys can be ignored by Lemma~\ref{lem:extv}.
If only one of these options is viable, the candidate is placed accordingly.
If both left and right are possible, we place the candidate arbitrarily right.
If neither is possible, the preference profile is not single-peaked.

Testing whether a vote $V_k$ imposes restrictions on the placement of a candidate is achieved by four conditions.
These conditions distinguish four categories of candidates:
candidates in $A_R$, candidates in $A_L$, candidates that have not yet been placed ($C_{>i}=\{c_{i+1},\ldots,c_m\}$), and the candidate that is currently under consideration ($c_i$).
We are only checking for valleys that include $c_i$.
This gives rise to the following four conditions:
\eqref{eq:cond1} and \eqref{eq:cond2} test whether placing $c_i$ on the right-hand side leads to valleys,
\eqref{eq:cond3} and \eqref{eq:cond4} do the same for the left-hand side.
Figure~\ref{fig:cond} displays a graphical representation.
Note that we do not verify whether a v-valley arises with $a_{\ell}\succ c_i$ and $a_r\succ c_i$, where $a_{\ell}\in A_L$ and $a_r\in A_R$; such a valley would have already be detected at an earlier stage of the algorithm (cf.~correctness proof).

\begin{figure}
\centering
\includegraphics[scale=1]{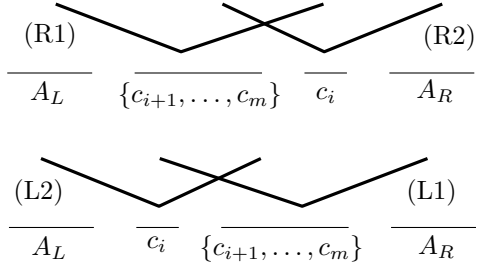}
\caption{Graphical representation of the conditions testing whether $c_i$ can be placed on the right-hand side (\ref{eq:cond1}, \ref{eq:cond2}) or on the left-hand side (\ref{eq:cond3}, \ref{eq:cond4})}
\label{fig:cond}
\end{figure}

Since we are working with weak orders, we do not have to consider every candidate triple possibly fulfilling these conditions but have to check only maximal or minimal candidates.
More specifically, checking whether there is a candidate $c\in A_L$ and $c'\in C_{>i}$ with $c\succ c'$ is equivalent to whether any maximal element in $A_L$ is preferred to some minimal element in $C_{>i}$ (subject to $\succ$).
For $k\in[n]$, let $\min_k(X)$ denote a function that picks an element in $X$ that is minimal with respect to $\succ_k$.
The function $\max_k(X)$ is defined analogously.
Now, we can formally define the four conditions:
\begin{align}
\tag{R1}\label{eq:cond1}c_i \succ_k \min_k(C_{>i}) &\text{\quad and \quad}
\max_k(A_L) \succ_k \min_k(C_{>i})  \\
\tag{R2}\label{eq:cond2}\max_k(C_{>i}) \succ_k c_i &\text{\quad and \quad} \max_k(A_R)\succ_k c_i\\
\tag{L1}\label{eq:cond3}c_i \succ_k \min_k(C_{>i}) &\text{\quad and \quad} \max_k(A_R)\succ_k \min_k(C_{>i}) \\
\tag{L2}\label{eq:cond4}\max_k(C_{>i}) \succ_k c_i &\text{\quad and \quad} \max_k(A_L)\succ_k c_i
\end{align}

Using these four definitions, we can give a succinct description of the Guided Algorithm,  (Algorithm~\ref{alg:untie}).
For a correctness proof of Algorithm~\ref{alg:untie}, we refer the reader to the appendix, Section~\ref{sec:app:corr-guided}.

{\renewcommand{\baselinestretch}{1.05}
\IncMargin{0.60em}
\begin{algorithm}
\Indm
\smallskip
  \KwIn{A set of candidates $C$, a preference profile of weak orders $\mathcal{P}=(V_1,\ldots V_n)$ including a guiding vote $c_m \succ c_{m-1} \succ \cdots \succ c_1$.}
  \KwOut{An axis $A$ or \nsp.}
\smallskip
\Indp
  \DontPrintSemicolon
  $A_L \leftarrow \left\langle\right\rangle$,
  $A_R \leftarrow \left\langle c_1\right\rangle$\;
  \For{$i \leftarrow 2\ldots m$} {
  $\textit{right} \leftarrow\true$; $\textit{left} \leftarrow\true$\;
  \For{$k \leftarrow 2\ldots n$} {
	  \If{condition \eqref{eq:cond1} or \eqref{eq:cond2} holds}	{  
	  	  $\textit{right} \leftarrow\false$\;
	  }
	  \If{condition \eqref{eq:cond3} or \eqref{eq:cond4} holds}	{  
	  	  $\textit{left} \leftarrow\false$\;
	  }
  }
  \eIf{$\textit{right} =\true$} {
    $A_R\leftarrow \left\langle c_i \rhd  A_R \right\rangle$\;
  } {
    \eIf{$\textit{left} =\true$} {
      $A_L\leftarrow \left\langle A_L \rhd  c_i \right\rangle$\;
  } {
    \Return{\nsp}
  }
  }
  }
  \Return{$A_L\rhd A_R$}\medskip
\caption{The Guided Algorithm}
\label{alg:untie}
\end{algorithm}
\DecMargin{0.6em}
} %

\addtocounter{example}{-1}
\begin{example}[{\bf continued}]
We apply the Guided Algorithm to
$\mathcal{P}=(V_1,V_2,V_3)$ with
$V_1=\langle a\succ b \succ c \sim d \succ e \rangle$,
$V_2=\langle a\sim b \sim c \succ d \succ e \rangle$, and
$V_3=\langle e\sim d \succ b \succ c \succ a \rangle$.
We use the implicitly found guiding order   $\langle a\succ b \succ c \succ d \succ e \rangle$
and thus start with $A_L = \left\langle\right\rangle$ and $A_R = \left\langle e\right\rangle$.
We now want to place $d$ on the axis.
None of the four conditions is satisfied for voters $V_1$ and $V_2$, but $V_3$ satisfies condition (L1).
Thus, $d$ has to be placed right, i.e., $A_R\leftarrow \left\langle d \rhd e\right\rangle$.
We continue with the third-lowest candidate on the guiding order: $c$.
Here, we encounter a problem: when considering $V_3$, we see that $c$ satisfies both (R2) and (L1). Thus, the algorithm cannot place $c$ and returns $\nsp$.
Using other implicitly given guiding orders, e.g., $\langle d \succ b \succ c \succ e \succ a \rangle$, would have yielded the same outcome.
\end{example}

Theorem~\ref{thm:guided-weak} claims that the Guided Algorithm requires $\bigO(m\cdot n)$ time.
Note that this is only possible if the conditions (\ref{eq:cond1}, \ref{eq:cond2}, \ref{eq:cond3}, \ref{eq:cond4}) can be checked in constant time.
Thus, the minima and maxima have to be computable in constant time.
For $\max_k(A_L)$ and $\max_k(A_R)$ this is easily possible by storing and updating these two values:
if $c_{i}$ is placed left, we update $\max_k(A_L)$ in case $c_{i}$ is the new maximum (with respect to $\succ_k$); if $c_{i}$ is placed right, we proceed analogously $\max_k(A_R)$.
For computing a minimal value of $C_{>i}$, observe that the set $C_{>i}$ becomes smaller with increasing $i$.
Thus, a minimal value of $C_{>i}$ might disappear at some point and a new (larger) value has to be chosen.
The new minimum is the smallest element (with respect to $\succ_k$) in $C_{>i}$ that is at least as large as the old minimum.
If we maintain pointers to the minimum elements, the amortized cost of this update procedure is $\bigO(1)$.
A maximal value of $C_{>i}$ can be found analogously.

\subsection{The Unguided Algorithm}
\label{sec:unguided}

Our next algorithm is applicable to top orders. As it is not dependent on a guiding vote, we refer to it as the Unguided Algorithm.
The Unguided Algorithm has a runtime of $\bigO(m^2\cdot n)$, the same as achieved by the consecutive ones approach. 
We see that from the perspective of worst-case complexity the Unguided Algorithm is inferior to the consecutive ones approach: while having the same worst-case runtime, it is applicable only to a smaller domain (top orders vs.\ weak orders).
However, a point made in favor of the Guided Algorithm also applies here: The main strength of the Unguided Algorithm is its relative simplicity. It can be seen as a solution method for certain consecutive ones instances and---due to this specialization---does not require the full power of consecutive ones machinery. It is thus easier to implement and can be used in characterization proofs.\footnote{Preliminary work with Dominik Peters indicates that the Unguided Algorithm can be used to characterize the single-peaked-or-caved domain: A total order $A$ is single-peaked-or-caved with respect to $A$ if either $T$ or its reverse is single-peaked with respect to $A$. This domain is clearly more general than the single-peaked domain, but less general than the single-peaked-on-a-circle domain \citep{peters2016spoc}.}

In the description of the algorithm we assume that the input preference profile is \emph{connected}.
Let us consider a simple graph with vertices corresponding to candidates.
We connect two vertices with an edge whenever the corresponding two candidates are both ranked in some top order.
A profile of top orders is called connected if this graph has only one connected component.
This assumption does not limit the applicability:
if two or more connected components exist in this graph, we can use the algorithm for each component (i.e., its respective candidates and votes) and concatenate the resulting axes in arbitrary order.
\IncMargin{0.60em}
{\renewcommand{\baselinestretch}{1.05}
\begin{algorithm}[t]
  \SetKwFunction{GuidedSP}{GuidedSP}
  \SetKwFunction{IntersectingVote}{IntersectingVote}
  \SetKwInOut{Input}{input}
  \SetKwInOut{Output}{output}
\smallskip
\Indm
  \KwIn{A set of candidates $C$ and a connected preference profile of top orders $\mathcal{P}=(V_1,\ldots V_n)$.}
  \KwOut{An axis $A$ or \textit{not\_single\_peaked}.}
\Indp
  \smallskip
  \DontPrintSemicolon
  \ForEach{$c_\text{start}\in C$\nllabel{line:cstart}}{
  	$A\leftarrow \left\langle c_\text{start} \right\rangle $\;
  	\For{$i \leftarrow 1 \ldots m$\nllabel{line:fori}} {
	  	\ForEach{vote $V\in\mathcal{P}$ that has $a_i$ as its top-ranked candidate\nllabel{line:foreachvotewithpeaki}} {
			\If{$A\oplus V = \ $\incompatible\nllabel{line:oplus}}{
				Continue with next $c_\text{start}\in C$ in line~\ref{line:cstart}.\;
			} 
			\lElse{
				$A \leftarrow A \oplus V$\nllabel{line:end:foreachvotewithpeaki}
			}
		} %
		\If{$\card{A} = i$ and $i<m$\nllabel{line:nomorepeakinA}} {
			$V \leftarrow $\IntersectingVote{$A$}\nllabel{line:intersectingvote}\;
			\If{$a_i\notin V$\nllabel{line:ainotinV1}}
			{Continue with next $c_\text{start}\in C$ in line~\ref{line:cstart}.\nllabel{line:ainotinV2}}
			Let $x$ be a new candidate, distinct from those in $C$.\;
			$C'\leftarrow \{c\in V \mid c \succ a_i\}\cup\{a_i,x\}$\;
			\For{$k \leftarrow 1 \ldots n$\nllabel{line:fork}} {
				$V_k'\leftarrow$ \textnormal{\texttt{RepTop}}$(V_k,C\setminus(A\cup C'),x)$\;
			}
			$\mathcal{P}'\leftarrow (V_1',\dots,V_n')[C']$\nllabel{line:end:fork}\;
		  	$A'\leftarrow $\ axis returned by Guided Alg.\ with input $(\mathcal{P}',V[C'],a_i,x)$\;
		  	\If{$A'=\,$\nsp} {
			  	Continue with next $c_\text{start}\in C$ in line~\ref{line:cstart}.\;
	  		}
	  		\lElse{
	  			$A\leftarrow  A \rhd  A'[C'\setminus\{x\}]$\nllabel{line:end:nomorepeakinA}
		  	}
		}
  	}
  	\Return{$A$}\; %
  }
  \Return{\nsp}
  \bigskip
\caption{The Unguided Algorithm}\label{alg:unguided}
\end{algorithm}
} %
\DecMargin{0.6em}

The Unguided Algorithm (Algorithm~\ref{alg:unguided}) works as follows:
First, we choose a candidate $c_\text{start}$ which is going to be the leftmost candidate on the axis $A$.
Since we have no guiding vote, each candidate might be placed at the leftmost position.
Hence we loop over all candidates (line~\ref{line:cstart}).
The corresponding axis under construction is $A = \left\langle c_\text{start} \right\rangle $.
We now aim to complete this axis by adding candidates to the right in such a way that all votes are single-peaked with respect to this axis.
To this end we employ the loop in line~\ref{line:fori}.
In this loop (variable $i$) we infer from the already placed candidate $a_i$ (the $i$-th candidate on $A$ from left) the candidate $a_{i+1}$ (or even more candidates farther to the right), or infer that $A$ cannot be extended to a single-peaked axis and thus try another start candidate.

Lines~\ref{line:foreachvotewithpeaki} to~\ref{line:end:foreachvotewithpeaki} are based on the following observation:
Let us assume that at a certain point $A=\left\langle c_1 \rhd  c_2 \rhd  c_3  \right\rangle$ and $V = \left\langle c_3 \succ c_2 \succ c_4 \succ c_5  \succ \bullet \right\rangle  \in \mathcal{P}$.
Since $c_3$, the peak of $V$, is already contained in $A$, there is only one compatible extension of $A$: $\left\langle c_1 \rhd  c_2 \rhd  c_3 \rhd  c_4 \rhd  c_5  \right\rangle$.
We formalize this extension operation with the $\oplus$ operator:
\begin{definition}
Let $A$ be an incomplete axis and $V$ a top order.
Furthermore, let $V[C\setminus A]=\left\langle c'_1 \succ c'_2 \succ \ldots \succ c'_j \succ \bullet\right\rangle$.
We define $A\oplus V = \left\langle A \rhd  c'_1 \rhd  c'_2 \rhd  \ldots \rhd  c'_j \right\rangle$ if $V$ is possibly single-peaked with respect to this axis and $A\oplus V = \incompatible$ otherwise.\label{def:oplus}
\end{definition}
The loop in line~\ref{line:foreachvotewithpeaki} considers all votes $V$ that have candidate $a_i$ as their top-ranked candidate (i.e., their maximal element).
If $A\oplus V = \incompatible$ then $A$ cannot be extended to a complete, single-peaked axis and we consider the next $c_\text{start}\in C$ in line~\ref{line:cstart}.
Otherwise, we obtain a new incomplete axis $A \leftarrow A\oplus V$.

It might be the case that the $({i+1})$st %
candidate on $A$ has not yet been determined after these steps.
The lines~\ref{line:nomorepeakinA} to~\ref{line:end:nomorepeakinA} deal with this case.
As the input profile is connected, there has to be at least one vote that ranks both a candidate on $A$ and a candidate that has not been placed yet.
The procedure \texttt{IntersectingVote} in line~\ref{line:intersectingvote} returns such a vote $V$ with $A\cap V\neq \emptyset$ and $V\setminus A\neq\emptyset$.
This procedure can be efficiently precomputed in such a way that it requires only $\bigO(m)$ time to provide an answer; details can be found in the proof of Theorem~\ref{thm:unguided}.

Let $V$ be a vote returned by \texttt{IntersectingVote}.
It holds that $V$'s top-ranked candidate is not placed on $A$ yet:
If its top-ranked candidate were contained in $A$, then $V$ would have been already considered in the first part of the algorithm (lines~\ref{line:foreachvotewithpeaki} to \ref{line:end:foreachvotewithpeaki}).
If $V$ does not contain $a_i$ (and thus $a_i$ is ranked last in $V$), $A$ cannot be extended to a single-peaked axis (lines~\ref{line:ainotinV1} and~\ref{line:ainotinV2}).
Now,
we employ the Guided Algorithm of Section~\ref{sec:guided} to find a further extension of $A$.
The main idea is to use $V$ as a guiding vote and find an axis for the candidates in $\{c\in V \mid c \succ a_i\}$.
In principle, this axis can be found independently of the existing axis $A$.
However, the leftmost and rightmost candidates have to be chosen with regard to ``external'' considerations:
The leftmost candidate has to be $a_i$, otherwise $A$ and the newly obtain partial axis $A'$ could not be combined.
For the rightmost candidate, we have to consider votes with candidates placed on the axis in future steps.
The following example illustrates the issue.
\begin{example}\label{ex:unguided}
Let $A=\left\langle a  \right\rangle$, $V_1 = \left\langle b \succ c \succ a \succ \bullet\right\rangle$ and $V_2 = \left\langle c \succ d \succ \bullet\right\rangle$.
The vote $V_1$ intersects $A$ and hence $C'=\{a,b,c\}$.
We employ the Guided Algorithm and might obtain $A'=\left\langle a \rhd  c \rhd  b  \right\rangle$.\footnote{Whether we obtain this axis or $\left\langle a \rhd  b \rhd  c  \right\rangle$ depends on whether the algorithm prefers placing candidates to the left or to the right if both choices are possible.}
Now observe that $A\oplus A' = A'$ can no longer be extended in a way that it is single-peaked for $V_2$.
This would have been possible if $c$ had been chosen as the rightmost candidate in $A'$.
\end{example}
As we can see from this example, we sometimes have to ``force'' the rightmost candidate in~$A'$.
We do this by adding an additional candidate $x$ to every vote (line~\ref{line:fork} to~\ref{line:end:fork}).
Hence, our candidate set under consideration is at this point $C'\leftarrow \{c\in V \mid c \succ a_i\}\cup\{a_i,x\}$.
The new candidate $x$ is now placed in each vote at the position of the highest-ranked candidate among those not contained in $A \cup C'$.
This is done by the \textnormal{\texttt{RepTop}} function:
\textnormal{\texttt{RepTop}}$(V,D,x)$ finds the one candidate in vote $V$ that is the highest-ranked among the candidates in $D$ and replaces it with candidate $x$.
By forcing this element $x$ to be the rightmost candidate, we ensure that $A'$ is chosen with consideration to all votes with ranked candidates not in $C'$.

\addtocounter{example}{-1}
\begin{example}[{\bf continued}] We apply \textnormal{\texttt{RepTop}}$(V,D,x)$ to the votes $V_1$ and $V_2$ with candidate sets $C'=\{a,b,c,x\}$ and $D=\{d\}$.
We obtain the votes $V'_1[C']=\left\langle b \succ c \succ a \succ x\right\rangle$ and $V_2'[C'] = \left\langle c \succ x \succ \bullet \right\rangle$.
Now, the Guided Algorithm can only return the axis $\left\langle a  \rhd b \rhd  c \rhd  x \right\rangle$.
\end{example}

We obtain a profile $\mathcal{P}'$ on the candidate set $C'$.
We now run the Guided Algorithm with input $(\mathcal{P}',V[C'],A',a_i,x)$, i.e., we employ the Guided Algorithm for the profile $\mathcal{P}'$ and guiding vote $V[C']$.
Furthermore, we require that the leftmost candidate on the axis is $a_i$ and the rightmost is $x$.
The Guided Algorithm either returns \nsp or an axis $A'$.
If it returns \nsp, the next  $c_\text{start}\in C$ is considered.
Otherwise, we continue with the extended axis $A \leftarrow A\oplus A'[C'\setminus\{x\}]$, i.e., candidate $x$ is not placed on $A$.

\newcommand{\thmunguidedtext}{The \textsc{Top Order Single-peaked Consistency} problem can be solved in {$\bigO(m^2\cdot n)$} time.}
\begin{theorem}
\thmunguidedtext\label{thm:unguided}
\end{theorem}

We refer the reader to the appendix, Section~\ref{sec:app:unguided}, for the correctness proof of Algorithm~\ref{alg:unguided} and its runtime calculations.

\begin{example}
Let us illustrate the Unguided Algorithm with a full example. Let $\mathcal{P}=(V_1$, $V_2$, $V_3$, $V_4)$ with
$V_1 = \left\langle b \succ c \succ a \succ \bullet\right\rangle$ (as in Example~\ref{ex:unguided}),
 $V_2 = \left\langle c \succ d \succ \bullet\right\rangle$  (also as in Example~\ref{ex:unguided}),
$V_3 = \left\langle f \succ g \succ h \succ e \succ a \succ\bullet\right\rangle$, and
$V_4=\langle h\succ g \succ f \succ \bullet \rangle$.
We assume that the algorithm starts with $c_\text{start}=h$, as this choice leads to a successful run.
We have exactly one vote with $h$ has its top-ranked candidate: $V_4$.
Thus, $A\leftarrow A\oplus V_4 = \langle h \rangle \oplus \langle h\succ  g \succ f \succ \bullet \rangle = \langle h\rhd  g \rhd f \rangle$.
We now look for votes with $f$ as their top-ranked choice; there is again only one: $V_3$.
As before, $A\leftarrow A\oplus V_3 = \langle h\rhd  g \rhd f \rangle \oplus \left\langle f \succ g \succ h \succ e \succ a \succ\bullet\right\rangle =\langle h\rhd  g \rhd f \rhd e \rhd a\rangle$. This is by Definition~\ref{def:oplus}, since $V_3$ is possibly single-peaked with respect to $\langle h\rhd  g \rhd f \rhd e \rhd a\rangle$.
We now have reached the situation described in Example~\ref{ex:unguided}.
Thus, using the Guided Algorithm, we complete the axis to $A=\langle h\rhd  g \rhd f \rhd e \rhd a\rhd b \rhd c\rangle$ and see that $\mathcal{P}=(V_1,V_2,V_3,V_4)$ is possibly single-peaked with respect to this so found axis.
\end{example}

\subsection{A 2-SAT-Based Algorithm}
\label{sec:2sat}

Theorem~\ref{thm:guidedNP}, Theorem~\ref{thm:weaksp}, and  Theorem~\ref{thm:guided-weak} leave open the case of profiles of local weak orders which contain at least one total order.
Here, we show that this case is also polynomial-time solvable.

\begin{theorem}
If the given profile contains a total order, the \textsc{Local Weak Order Single-peaked Consistency} problem can be solved in $\mathcal{O}(n\cdot m^3)$ time.
\label{thm:2sat}
\end{theorem}

We encode a \textsc{Local Weak Order Single-peaked Consistency} instance in a \textsc{2-SAT} instance.
The \textsc{2-SAT} problem asks whether a Boolean formula in conjunctive
normal form (a conjunction of disjunctions of literals) where each clause has
size two (e.g., $(a\vee b)\wedge (\neg a \vee c)$) is satisfiable.
Solving \textsc{2-SAT} requires only linear time~\citep{aspvalllinear-time1979}.
The Boolean variables in our instance correspond to pairs of candidates, i.e., for each $a,b\in C$ we have a variable $ab$.
The intended meaning of these variables is that $ab=\true$ if and only if $a$ is left of $b$ on the axis.
Now, for each vote $V$ and triple $a,b,c\in C$, if $a\succ b$ and $c\succ b$ ($a,b,c$ may form a v-valley), we add the clauses 
\begin{align}
&(ba\vee cb)\text{\quad and}\label{eq:2sat-first-1}\\
&(ab\vee bc)\label{eq:2sat-first-2}
\end{align}
 to the \textsc{2-SAT} instance.
These clauses correspond the requirement that $b$ must not be placed between $a$ and $c$.
Finally, we add for each pair of variables $a,b$ the clauses 
\begin{align}
(ab\vee ba)\wedge(\neg ab\vee \neg ba),\label{eq:2sat-second}
\end{align}
corresponding to an exclusive or operator.
Solving the \textsc{2-SAT} instance either yields the information that the instance is not satisfiable or a true/false assignment to the variables.
In the first case, the profile is not single-peaked (as shown in Lemma~\ref{lem:2sat-corr1}).
In the second case, we obtain a relation $A = \{(a,b): ab=\true\}$ which is a total order and our desired axis (as shown in Lemma~\ref{lem:2sat-corr2}).
Since the instance contains at most $\mathcal{O}(n\cdot m^3)$ clauses, we obtain the stated runtime.

\begin{lemma}
If $\mathcal{P}$ is single-peaked with respect to an axis $A$, then the corresponding \textsc{2-SAT} instance is satisfiable.
\label{lem:2sat-corr1}
\end{lemma}
\begin{proof}
We construct a valid truth assignment as follows:
If $a\rhd b$ on $A$, then $ab=\true$ and $ba=\false$.
Clauses of the form \eqref{eq:2sat-first-1} and \eqref{eq:2sat-first-2} are satisfied because no v-valleys occur on~$A$; Clauses of the form \eqref{eq:2sat-second} are satisfied because v-valleys cannot occur ($\mathcal{P}$ is single-peaked with respect to an axis $A$).
\end{proof}

\begin{lemma}
The relation $A = \{(a,b): ab=\true\}$, as returned by the \textsc{2-SAT} algorithm, is a total order and $\mathcal{P}$ is single-peaked with respect to $A$.
\label{lem:2sat-corr2}
\end{lemma}
\begin{proof}
First, we want to show that $A$ is a total order.
Asymmetry and totality follow immediately from~\eqref{eq:2sat-second}. 
Towards a contradiction assume that $A$ is not transitive, i.e., there exist three candidates $x,y,z$ such that $\{(x,y),(y,z),(z,x)\}\subseteq A$.
Thus, $xy=yz=zx=\true$ and $yx=zy=xz=\false$.
Let $V$ be a total order contained in $\mathcal{P}$ (there exists at least one).
We distinguish three cases:
\begin{itemize}
\item The last ranked candidate of $x,y,z$ in $V$ is $y$:
By~\eqref{eq:2sat-first-1}, it has to hold that $(yx \vee zy)$, which is not the case.
\item The last ranked candidate of $x,y,z$ in $V$ is $x$:
By~\eqref{eq:2sat-first-2}, it has to hold that $(yx\vee xz)$, which is not the case.
\item The last ranked candidate of $a,b,c$ in $V$ is $c$:
By~\eqref{eq:2sat-first-1}, it has to hold that $(zy\vee xz)$, which is not the case.
\end{itemize}
Thus, $A$ is transitive.
It remains to show that $\mathcal{P}$ is single-peaked with respect to $A$.
Assume that there is a valley $a\succ b$, $c\succ b$ in some vote and it holds that $\{(a,b),(b,c)\}\subseteq A$.
Due to this valley, our \textsc{2-SAT} instance contains the clause $(ba\vee cb)$.
Thus, $(b,a)\in A$ or $(c,b)\in A$ and $ba=\true$ or $cb=\true$.
Hence, $ab=\false$ or $bc=\false$, which contradicts our assumption that $\{(a,b),(b,c)\}\subseteq A$.
\end{proof}

\section{Other Single-Peaked Concepts for Weak Orders}
\label{sec:other}

We now turn to three other models of single-peaked preferences for weak orders---Black single-peaked
and single-plateaued preferences, and at the end of the section, necessarily single-peaked preferences.

For weak orders, it is more common to view $a\sim b$ as indifference instead of missing information (i.e., incompleteness).
From the algorithmic point of view, which we had taken so far, this distinction is irrelevant.
From a conceptional point of view, it clearly makes a difference, in particular for the underlying assumption of voters' true preferences.
It may even be that both incompleteness and indifference are simultaneously be present in a voter's preference information; we however disregard this possibility.

Let us start with the definition of single-peaked preferences as introduced by Black, which applies to weak orders as well. %

\begin{definition}\label{def:blacksp}
A preference profile $\mathcal{P}$ of weak orders is \emph{Black single-peaked} with
respect to an axis $A$ if for every $V \in \mathcal{P}$,
$A$ can be split
at the most-preferred candidate (peak) of $V$ into two segments
$X$ and $Y$ (one of which can be empty) such that
$V$ has strictly increasing preferences along $X$ and strictly decreasing preferences along $Y$.
\end{definition}

In other words, for a weak preference order to be Black single-peaked it must have a single
most-preferred candidate and can
only contain indifference between at most two candidates at each
position in the order. Otherwise the segments $X$ and $Y$ referred to in
Definition~\ref{def:blacksp} would not be \emph{strictly}
increasing/decreasing.

A slightly weaker restriction than Black single-peakedness for weak 
orders %
is the single-plateaued restriction~\cite[Chapter 5]{bla:b:polsci:committees-elections}, which
extends Black single-peakedness to allow voters to state multiple
most-preferred candidates (i.e., a single plateau instead of a single peak).\footnote{We note that single-plateaued preferences are occasionally referred to as single-peaked preferences, e.g.,~in~\cite[Chapter 9]{fis:b:theory}.}
The model of possibly single-peaked preferences considered in this paper generalizes both single-plateaued and Black single-peaked preferences. Examples and visualizations of these three restrictions can be found in Figure~\ref{fig:weak-orders-overview}.

\begin{figure}
	\centering
	\begin{subfigure}{0.3\textwidth}
		\centering
		\begin{tikzpicture}[yscale=0.5,xscale=0.65]
		
		\def\xmin{1}
		\def\xmax{7}
		\def\ymin{0}
		\def\ymax{7}
		
		\draw[step=1cm,black!20,very thin] (\xmin,\ymin) grid (\xmax,\ymax);

		\draw[->] (\xmin -0.5,\ymin) -- (\xmax+0.5,\ymin) node[right] {};
		\foreach \x/\xtext in {1/a, 2/b, 3/c, 4/d, 5/e, 6/f, 7/g}
		\draw[shift={(\x,\ymin)}] (0pt,2pt) -- (0pt,-2pt) node[below] {$\strut\xtext$};
		\foreach \x/\xtext in {1, 2,3,4,5,6}
		\node[below] at (\x+0.5,\ymin) {$\strut\rhd$};  
		
		\foreach \x/\y in {4/6,5/6,3/4, 6/6, 2/3, 7/2, 1/1}
		\node[fill=blue, circle, inner sep=0.6mm] at (\x,\y) {};
		
		\draw[thick,blue] (1,1)--(2,3)--(3,4)--(4,6) -- (5,6)--(6,6)--(7,2);
		
		\foreach \x/\y in {4/4,5/3, 6/2, 3/5, 2/5, 7/1, 1/5}
		\node[fill=green!50!black, circle, inner sep=0.6mm] at (\x,\y) {};
		
		\draw[thick,green!50!black] (1,5)--(2,5)--(3,5)--(4,4) -- (5,3)--(6,2)--(7,1);  
		
		\foreach \x/\y in {1/2,2/4,3/7, 4/7, 5/5, 6/3, 7/1}
		\node[fill=red!50!black, circle, inner sep=0.6mm] at (\x,\y) {};
		
		\draw[thick,red!50!black]  (1,2)--(2,4)--(3,7)--(4,7) -- (5,5)--(6,3)--(7,1);
		
		\end{tikzpicture}
		\caption{Single-Plateaued}
	\end{subfigure}
	\quad
	\begin{subfigure}{0.3\textwidth}
		\centering
		\begin{tikzpicture}[yscale=0.5,xscale=0.65]
		
		\def\xmin{1}
		\def\xmax{7}
		\def\ymin{0}
		\def\ymax{7}
		
		\draw[step=1cm,black!20,very thin] (\xmin,\ymin) grid (\xmax,\ymax);

		\draw[->] (\xmin -0.5,\ymin) -- (\xmax+0.5,\ymin) node[right] {};
		\foreach \x/\xtext in {1/a, 2/b, 3/c, 4/d, 5/e, 6/f, 7/g}
		\draw[shift={(\x,\ymin)}] (0pt,2pt) -- (0pt,-2pt) node[below] {$\strut\xtext$};
		\foreach \x/\xtext in {1, 2,3,4,5,6}
		\node[below] at (\x+0.5,\ymin) {$\strut\rhd$};  
		
		\foreach \x/\y in {4/6,5/7,3/4, 6/6, 2/3, 7/3, 1/1}
		\node[fill=blue, circle, inner sep=0.6mm] at (\x,\y) {};
		\draw[thick,blue] (1,1)--(2,3)--(3,4)--(4,6) -- (5,7)--(6,6)--(7,3);
		
		\foreach \x/\y in {4/4,5/3, 6/2, 3/5, 2/6, 7/1, 1/2}
		\node[fill=green!50!black, circle, inner sep=0.6mm] at (\x,\y) {};
		\draw[thick,green!50!black] (1,2)--(2,6)--(3,5)--(4,4) -- (5,3)--(6,2)--(7,1);  
		
		\foreach \x/\y in {1/3, 2/4, 3/6, 4/7, 5/5, 6/3, 7/2}
		\node[fill=red!50!black, circle, inner sep=0.6mm] at (\x,\y) {};
		\draw[thick,red!50!black]  (1,3)--(2,4)--(3,6)--(4,7) -- (5,5)--(6,3)--(7,2);
		
		\end{tikzpicture}
		\caption{Black Single-Peaked}
	\end{subfigure}
	\quad
	\begin{subfigure}{0.3\textwidth}
		\centering
		\begin{tikzpicture}[yscale=0.5,xscale=0.65]
		
		\def\xmin{1}
		\def\xmax{7}
		\def\ymin{0}
		\def\ymax{7}
		
		\draw[step=1cm,black!20,very thin] (\xmin,\ymin) grid (\xmax,\ymax);

		\draw[->] (\xmin -0.5,\ymin) -- (\xmax+0.5,\ymin) node[right] {};
		\foreach \x/\xtext in {1/a, 2/b, 3/c, 4/d, 5/e, 6/f, 7/g}
		\draw[shift={(\x,\ymin)}] (0pt,2pt) -- (0pt,-2pt) node[below] {$\strut\xtext$};
		\foreach \x/\xtext in {1, 2,3,4,5,6}
		\node[below] at (\x+0.5,\ymin) {$\strut\rhd$};  
		
		\foreach \x/\y in {4/6,5/7,3/2, 6/6, 2/2, 7/2, 1/1}
		\node[fill=blue, circle, inner sep=0.6mm] at (\x,\y) {};
		\draw[thick,blue] (1,1)--(2,2)--(3,2)--(4,6) -- (5,7)--(6,6)--(7,2);
		
		\foreach \x/\y in {4/2,5/2, 6/2, 3/5, 2/5, 7/1, 1/2}
		\node[fill=green!50!black, circle, inner sep=0.6mm] at (\x,\y) {};
		\draw[thick,green!50!black] (1,2)--(2,5)--(3,5)--(4,2) -- (5,2)--(6,2)--(7,1);  
		
		\foreach \x/\y in {1/4,2/4,3/7, 4/7, 5/5, 6/3, 7/3}
		\node[fill=red!50!black, circle, inner sep=0.6mm] at (\x,\y) {};
		\draw[thick,red!50!black]  (1,4)--(2,4)--(3,7)--(4,7) -- (5,5)--(6,3)--(7,3);
		
		\end{tikzpicture}
		\caption{Possibly Single-Peaked}
	\end{subfigure}
\caption{Overview of Single-Peakedness Models for Weak Orders}\label{fig:weak-orders-overview}
\end{figure}
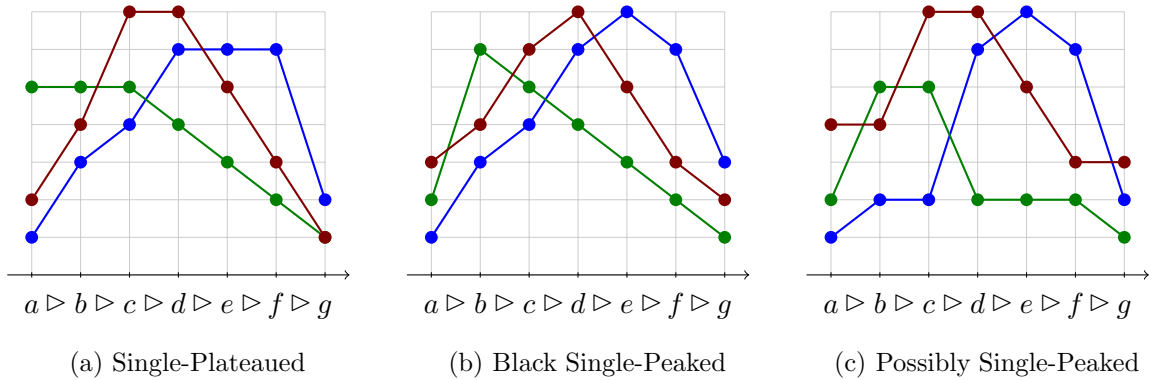

In this section, we have three main goals:
In Section~\ref{subsec:trans} we differentiate single-peakedness models with respect to the axiomatic properties they guarantee.
In Sections~\ref{subsec:plateued} and~\ref{subsec:black} we 
show that the consecutive ones approach can also be used to solve the consistency problem for single-plateaued and
for Black single-peaked preferences.
Finally, in Section~\ref{subsec:necessary}, we briefly consider the concept of \emph{necessarily} single-peaked preferences.

\subsection{Transitive Majority Relations and Condorcet Winners}\label{subsec:trans}

Preference profiles that satisfy either the Black single-peaked or single-plateaued restriction have several
desirable properties.
One of them is that both Black single-peaked and single-plateaued preferences guarantee transitive majority
relations \citep{bla:j:rationale-of-group-decision-making,bla:b:polsci:committees-elections}.
To be precise, we say that candidate $a$ is preferred to $b$ by a (strict) majority if $|V\in \mathcal{P}:a\succ b| > |V\in \mathcal{P}:b\succ a|$; in this case we write $a >_m b$, where
$>_m$ denotes the majority relation.
A transitive majority relation further guarantees the existence of one or more weak Condorcet winners.
A candidate $a$ is a weak Condorcet winner if there is no candidate $b$ with $b >_m a$.
Upon closer inspection, the Black single-peaked and the single-plateaued  restriction show different characteristics; we refer the reader to \citet{bar:j:indifference-domain}, who further discusses how the amounts of indifference permitted in these restrictions impact their properties.

The gain in generality for possibly single-peaked preferences comes at a price: This condition is no longer restrictive enough to guarantee a transitive majority relation.
This can be seen by considering the following simple profile reproduced from \citeA[Table~9.1]{fis:b:theory}.
Let $\mathcal{P}=(V_1,V_2,V_3,V_4,V_5)$ with

\[\begin{matrix}
  V_1:   & \left\langle b \succ a \succ c\right\rangle\\
  V_2,V_3: & \left\langle c \succ b \succ a\right\rangle\\
  V_4,V_5: & \left\langle a \succ b \sim c\right\rangle\\
\end{matrix}\]

The corresponding majority relation satisfies $a>_m c >_m b  >_m a $, and thus is clearly not transitive. Furthermore, no Condorcet winner exists.
Note that $\mathcal{P}$ is possibly single-peaked with respect to the axis
$a \rhd b \rhd c$ and that this profile is comprised of top orders. Thus,
there exist preference profiles of top orders that are possibly single-peaked
and do not have a transitive majority relation nor a Condorcet winner.

\subsection{Single-Plateaued Profiles and Their Recognition}\label{subsec:plateued}

Single-plateaued preferences are a much more restrictive model than
possibly single-peaked 
preferences since they are essentially Black single-peaked except that
each preference order can 
have multiple most-preferred
candidates~\cite[Chapter 5]{bla:b:polsci:committees-elections}.

We first provide a characterization of single-plateaued preferences, similar to the v-valley based characterization of possibly single-peaked preferences (Lemma~\ref{lem:extv}).
Since a preference order must be strictly increasing and then %
strictly decreasing with respect to an axis %
(except for its most-preferred candidates), we can again use
the v-valley substructure. Furthermore, we will need
an additional substructure to prevent two candidates that are ranked indifferent in a voter's preference order from appearing on the same side of that voter's peak (plateau).

\begin{definition}\label{def:nonpeak-plat}
A preference order $V$ contains a
\emph{plateau with respect to an axis $A$} if there exist
candidates $a,b \in C$ such that $a$ and $b$ are adjacent
in $A$ and $V$ states $a \sim b$.
A preference order $V$ contains a
\emph{nonpeak plateau with respect to an axis $A$} if there exist
candidates $a,b,c, \in C$ such that
$a \rhd b \rhd c$ on $A$ and $V$ states either $a \succ b \sim c$ or $c \succ b \sim a$.
\end{definition}

We can now state the following analogue of Lemma~\ref{lem:extv}.

\begin{lemma}\label{lem:plat}
Let $\mathcal{P}=(V_1,\ldots,V_n)$ be a profile of weak orders.
The following two statements are equivalent.
\begin{enumerate}
\item[(i)] The profile $\mathcal{P}$ is single-plateaued with respect to $A$.
\item[(ii)] Every vote $V\in\mathcal{P}$ contains neither a v-valley nor a nonpeak plateau
    with respect to $A$.
\end{enumerate}
\end{lemma}

\begin{proof}
Let $\mathcal{P}$ be a preference profile of weak orders, and let
$A$ be an axis.

If $\mathcal{P}$ is single-plateaued with respect to $A$ then for every preference order
$V \in \mathcal{P}$, $A$ can be split into segments $X$, $Y$, and $Z$ such that $V$
is strictly increasing along $X$, remaining constant along $Y$, and
strictly decreasing along $Z$.
Since $v$ is only ever strictly decreasing along $Z$ and $Z$ is the
rightmost segment of $A$, $V$ cannot contain a v-valley with respect to $A$.
For a nonpeak plateau to exist with respect to $A$ there must exist
candidates $a,b,c \in C$ such that $a \rhd b \rhd c$ in $A$ and $V$ states either
$a \succ b \sim c$ or $c \succ b \sim a$.

We first consider the case where $V$ states $a \succ b \sim c$. Since $V$ strictly prefers
$a$ to $b$ and $a$ to $c$, and both $b$ and $c$ are to the right of $a$ on
the axis,
we know that both $b$ and $c$ must be in segment $Z$.
However, $V$ is strictly decreasing along $Z$, so $V$ cannot have a nonpeak
plateau of this form.

We now consider the case where $V$ states $c \succ b \sim a$. Since $V$ strictly prefers
$c$ to $a$ and $c$ to $b$, and both $a$ and $b$ are to the left of $c$ on
the axis, %
we know that both $a$ and $b$ must be in segment $X$.
However, $V$ is strictly increasing along
$X$, so $V$ cannot have a nonpeak plateau of this form.

For the other direction, assume that no preference
order $V \in \mathcal{P}$ contains a v-valley with respect to $A$
and no preference order $V \in \mathcal{P}$
contains a nonpeak plateau with respect to $A$.
Since no preference order $V \in \mathcal{P}$ contains a v-valley with respect to $A$, we know from
Lemma~\ref{lem:extv} that $V$ is possibly single-peaked with respect to $A$. 
Since we also know that no preference order $V \in \mathcal{P}$ contains a nonpeak plateau
with respect to $A$ it is easy to see that $V$ is single-plateaued with respect
to $A$.~\end{proof}

We now extend the consecutive ones approach from Section~\ref{sec:consones} to single-plateaued profiles.
Since the nonpeak plateau substructure is needed in addition to the
v-valley substructure, we %
modify Construction~\ref{con:mat} so that
if a
preference order contains a nonpeak plateau with respect to an axis $A$,
then when the columns of its corresponding preference matrix are permuted according to
$A$ the matrix will contain a row with the sequence
$\cdots 1 \cdots 0 \cdots 1 \cdots$.
Further notice that if a preference order ranks three 
candidates indifferent to each other below its peak (plateau)
that it will have a nonpeak plateau with respect to \emph{every} possible axis.
To handle this case in the extension to
Construction~\ref{con:mat} we need to ensure that its corresponding preference
matrix will contain a row with the sequence $\cdots 1 \cdots 0 \cdots 1 \cdots$ for
every permutation of its columns.

\begin{construction}\label{con:mat3} %
Let $\mathcal P=(V_1,\dots,V_n)$ be a profile of weak orders over candidate set $C$ with $\card{C}=m$.
For each $V_i$ we initially construct an $m \times m$ binary matrix $X_i$
as described in Construction~\ref{con:mat}.

The following extensions to Construction~\ref{con:mat} ensure that
if $V_i$ has a nonpeak plateau with respect to an axis $A$ then when the columns
of $X_i$ are permuted according to $A$ it will not have the consecutive ones property.

If there exist
three
candidates $a,b,c \in C$ such that $V_i$ states $a \sim b \sim c$
and they are not $V_i$'s most-preferred candidates, then
output a matrix that does not have the consecutive ones property.

Otherwise, 
for each pair of
candidates $a,b \in C$ such that (i) $V_i$ ranks $a \sim b$ and (i) $a$ and $b$ are not maximal in $V_i$, %
append three additional rows to the matrix $X_i$:
to the column corresponding 
to $a$ append $\begin{bmatrix}0 & 1 & 1 \end{bmatrix}'$,
to the column corresponding to $b$ 
append $\begin{bmatrix}1 & 1 & 0 \end{bmatrix}'$,
to each column corresponding to a candidate 
strictly preferred to $a$ and $b$ 
append $\begin{bmatrix}1 & 1 & 1 \end{bmatrix}'$,
and to each column corresponding to a remaining candidate
append $\begin{bmatrix}0 & 0 & 0 \end{bmatrix}'$.
After constructing a matrix $X_i$ for each $V_i$, the matrices $X_1, \dots, X_n$
are row-wise concatenated to yield a matrix $X$.
\end{construction}

Construction~\ref{con:mat} ensures that no preference order contains
a v-valley and the extensions made in Construction~\ref{con:mat3} ensure
that no preference order contains
a nonpeak plateau.
So the following theorem can be shown by 
an argument similar to the proof of Theorem~\ref{thm:weaksp}.
Now the presence of v-valleys \emph{or}
nonpeak plateaus, not just v-valleys, is equivalent to a row
containing the sequence $\cdots 1 \cdots 0 \cdots 1 \cdots$.

\begin{theorem}\label{thm:splat}
The \textsc{Weak Order Single-plateaued Consistency} problem can be solved
in $\bigO(m^2\cdot n)$ time.
\end{theorem}

\subsection{Black Single-Peaked Consistency}\label{subsec:black}

Recall that a preference order is Black single-peaked with respect to an axis $A$ if it is
strictly increasing to a single most-preferred candidate (peak) and then strictly
decreasing with respect to $A$, this is consequently a special case of single-plateaued preferences. So we again use the v-valley substructure to obtain a characterization result, but
as before we need an additional
substructure. 
Even if no preference order
has a v-valley with respect to $A$ it may not be Black single-peaked because it is
indifferent between two candidates on the same side of its peak or
has more than one most-preferred candidate.
We can handle the first condition just mentioned with the
nonpeak plateau substructure used in the previous section, but
the second condition requires us to forbid \emph{any} kind of plateau. 

\begin{lemma}\label{lem:sp}
Let $\mathcal{P}=(V_1,\ldots,V_n)$ be a profile of weak orders.
The following two statements are equivalent.
\begin{enumerate}
\item[(i)] The profile $\mathcal{P}$ is Black single-peaked with respect to $A$.
\item[(ii)] Every vote $V\in\mathcal{P}$ contains neither a v-valley nor a
    plateau with respect to $A$.
\end{enumerate}
\end{lemma}
We omit the proof, as the argument is very similar to the proof of Lemma~\ref{lem:plat}.

Now, we extend Construction~\ref{con:mat3} so that if a preference order
contains a plateau with respect to an axis $A$, then when the columns of
its preference matrix are permuted according to $A$ the matrix will contain
the sequence $\cdots 1 \cdots 0 \cdots 1 \cdots$. Since
Construction~\ref{con:mat3} already
ensures this for the case of nonpeak plateaus, our
extended construction below only needs to add a condition for plateaus that
contain the most-preferred candidates in a given preference order.

\begin{construction}\label{con:mat2}
Follow Construction~\ref{con:mat3} except add the following
condition while constructing a preference matrix $X_i$ for each
preference order $V_i \in {\cal P}$. %
If there exist two candidates
$a,b \in C$ such that $V_i$ states $a \sim b$
and they are $V_i$'s most-preferred candidates,
then output a matrix that does not have the consecutive ones property.
\end{construction}

When a preference order
has a unique most-preferred candidate and
is single-plateaued, it is clearly also Black single-peaked.
Construction~\ref{con:mat2}
ensures that no preference order contains
more than one most-preferred candidate the same way that
Construction~\ref{con:mat3} ensures that no preference order contains
three or more candidates that are all ranked indifferent to each other and
that are not the most
preferred candidates, since
this always results in a nonpeak plateau.
So the following theorem can be shown by essentially the same argument as for Theorem~\ref{thm:splat}, but using Lemma~\ref{lem:sp} instead of
Lemma~\ref{lem:plat}.

\begin{theorem}\label{thm::blacksp}
The \textsc{Weak Order Black Single-peaked Consistency} problem can be solved
in $\bigO(m^2\cdot n)$ time.

\end{theorem}

\subsection{Necessarily Single-Peaked Preferences}\label{subsec:necessary}

The central definition of this paper is that of possibly single-peaked preferences.
As mentioned in  Section~\ref{sec:sp}, one could also consider \emph{necessarily single-peaked} preferences, i.e., incomplete preferences for which every extension is single-peaked.
We conclude this section by considering this much more restrictive variant.
For weak orders, we can give a clear answer how these two concepts compare:
necessarily single-peaked preferences can be characterized as a subclass of single-plateaued preferences.

\begin{proposition}
A profile of weak orders is necessarily single-peaked with respect to an axis $A$ if and only if it is single-plateaued with respect to $A$ in such a way that plateaus have a size of at most two.
\end{proposition}

\begin{proof}
Let $\mathcal{P}$, a profile of weak orders, be single-plateaued with respect to an axis $A$ such that all plateaus have a size $\leq 2$.
It is easy to see that for each $V \in {\cal P}$ that every extension of $V$ to a total order cannot contain a v-valley with
respect to $A$.

For the other direction, 
we consider two cases: (i) $\mathcal{P}$ is not single-plateaued with respect to $A$ and (ii) $\mathcal{P}$ is single-plateaued with respect to axis $A$ but a $V\in \mathcal{P}$ contains a plateau of size larger than two.
In case~(i), we use Lemma~\ref{lem:plat}. There are two ways in which the single-plateaued property can be violated:
a vote $V$ contains a v-valley or a nonpeak plateau with respect to $A$.
If $V$ contains a v-valley, then none of its extensions is single-peaked with respect to $A$.
If $V$ contains a nonpeak plateau, then it is of the form $p\succ a\sim b$, where $p$ is a maximal element of $V$, with $a\rhd b\rhd p$ on $A$ (or reversed).
Let $V'$ be any extension of $V$ for which $p \succ' a\succ' b$ holds.
Then $V'$ is not single-peaked with respect to $A$, and hence $\mathcal{P}$ is not necessarily single-peaked with respect to $A$.
Now, consider case~(ii). Let $V\in\mathcal{P}$ with $a\sim b\sim c$, all of which are maximal elements in $V$, i.e., we have a plateau of size at least three.
Without loss of generality assume that $a\rhd b\rhd c$ on $A$.
Let $V'$ be any extension of $V$ that satisfies $a \succ' b$ and $c \succ' b$.
Then $V'$ contains a v-valley and thus $\mathcal{P}$ is not necessarily single-peaked with respect to $A$.
\end{proof}

Interestingly, this ``plateau of size at most two'' definition is exactly the way~\citet{arr:b:polsci:social-choice} defined single-peakedness, introducing a slight deviation from Black's original definition
\cite<as mentioned by>{dummett1961stability}.

\section{Discussion and Future Work}
\label{sec:disc}

We have analyzed the \textsc{$\mathcal{T}$ Single-peaked Consistency} problem for $\mathcal{T}\in \{$\textsc{Partial Order, Local Weak Order, Weak Order, Top Order, Total Order}$\}$.
An overview of the results is displayed in Table~\ref{tab:results}.
Despite the \NP-completeness of \textsc{Partial Order Single-peaked Consistency}, we have found four fast algorithms for plausible application scenarios.
The Guided Algorithm and the \textsc{2-SAT} based algorithm each require a guiding vote.
Such an order is likely to exist---at least implicitly---for large preference profiles.
In the case that no guiding order exists, the consecutive ones approach and the Unguided Algorithm are applicable.
Also, we have found that \textsc{Partial Order Single-peaked Consistency} is solvable in polynomial time if the axis is already part of the input.
This covers a large spectrum of possible scenarios of real-world preferences where our algorithms could be applied.

\begin{table}
\begin{center}
\begin{tabular}{l|ll}
$\mathcal{T}$ & general & with guiding vote \\ 
\hline 
\textsc{Partial Order} & NP-c~(Cor~\ref{cor:POSPC-npc}) & NP-c~(Thm~\ref{thm:guidedNP}) \\ 
\textsc{Local Weak O.} & NP-c~(Thm~\ref{thm:SPE-npc}) & $\mathcal{O}(m^3 \cdot n)$~(Thm~\ref{thm:2sat}) \\ 
\textsc{Weak Order} & $\calO(m^2 \cdot n)$~(Thm~\ref{thm:weaksp}) & $\calO(m \cdot n)$~(Thm~\ref{thm:guided-weak}) \\ 
\textsc{Top Order} & $\calO(m^2 \cdot n)$~(Thm~\ref{thm:unguided}) & $\calO(m \cdot n)$~(Thm~\ref{thm:guided-weak}) \\ 
\textsc{Total Order} & $\calO(m \cdot n)$ (\citet{esc-lan-ozt:c:single-peaked-consistency}, also Thm~\ref{thm:guided-weak}) & n/a
\end{tabular}
\end{center}
\caption{Overview of the complexity results and algorithms for \textsc{$\mathcal{T}$ Single-Peaked Consistency}}\label{tab:results}
\end{table}%

\subsection{Implementation} 
We implemented our consecutive ones-based algorithm for the recognition of possibly single-peaked preferences for weak orders (and the extensions to the construction for single-plateaued and Black single-peaked preferences). The code is written in Python 3.7 and is available at \url{https://github.com/zmf6921/incompletesp}.
We use a SAT solver for the consecutive ones instances.
(This is a much simpler and more reliable approach than implementing the available linear-time algorithms and still yields very fast runtimes.)
We ran the algorithm on all currently available profiles of weak orders available on PrefLib.org~\cite{preflib};
these are in total 379 toc-instances\footnote{At PrefLib, profiles of weak orders are referred to as ``tied-order complete (toc)'', c.f.\ \url{http://www.preflib.org/data/format.php#toc}.}.
We found that none of these instances are possibly single-peaked.
This highlights a disadvantage of the single-peaked domain: its poor robustness. Even a minor change in a single voter's preferences can violate the single-peaked condition. And although the possibly single-peaked domain is a more general restriction, this drawback remains. 
Given this fragility, the result is not surprising but leads us to ask the question of how {\em close} these profiles are to being possibly single-peaked.
We discuss this research direction further below.
On a more positive note, even large instances with $>1{,}000$ candidates and instances with $>25{,}000$ distinct votes could be solved in seconds 
(the maximum running time was 20.2 seconds, the average was 0.8 seconds, the median 0.02 seconds, processed on a 4GHz Intel i7 with 16 GB of RAM).

\subsection{Top Orders and Scoring Rules}
We would like to mention one particular application of our algorithms concerning a specific class of voting rules, so-called scoring rules.
Scoring rules are specified by a family of scoring vectors. For an election over $m$ candidates, a scoring rule uses
the corresponding $m$-candidate scoring vector $(\alpha_1,\ldots,\alpha_m)$ to determine the winner(s).
A vote $\left\langle c_1 \succ \dots \succ c_m\right\rangle$ gives $\alpha_1$ points to $c_1$, $\alpha_2$ points to $c_2$, etc.
The candidate(s) with the highest score win.
Often, scoring vectors of the type $(\alpha_1,\ldots,\alpha_k,0,\ldots,0)$ with $\alpha_1 \ge \dots \ \alpha_k > 0$ are considered.
For example, the voting in the Eurovision Song Contest uses the scoring vector $(12,9,8,7,6,5,4,3,2,1,0,\ldots,0)$~\citep{ginsburgh2008eurovision}.
For such scoring vectors, top orders (with $k$ ranked candidates) constitute \emph{full} information for determining
the winner.
It is therefore debatable whether the input of such voting rules may be considered to be given as a profile of total orders, as total orders contain problem-irrelevant information.
\citet{bra-bri-hem-hem:j:sp2} study the constructive coalitional weighted manipulation problem for scoring rules in single-peaked elections.
The authors consider the axis to be part of the input (for good reasons as explained in their paper).
The computation of such an axis with existing algorithms is only possible if preferences are specified by total orders and thus contain problem-irrelevant information.
If only relevant information is given, i.e., the input consists of top orders, algorithms such as those presented in this paper are required.

\subsection{Future Research}
Let us discuss some directions for future research.
When discussing our implementation and the outcome of our PrefLib analysis, we mentioned the lack of robustness of (possibly) single-peaked profiles.
This issue of robustness has been addressed by considering {\em nearly} single-peaked profiles, i.e., profiles that are close to being single-peaked according to some measure \citep{elkind2012clone,cor-gal-spa:c:spwidth,cor-gal-spa:c:spsc-width,fal-hem-hem:j:nearly-sp,bredereck2016nicelystructured,jair/ErdelyiLP17-nearlysp}. This line of work is limited to profiles of total orders; an extension to possibly single-peaked profiles could further generalize the applicability of single-peakedness.
The first step in this direction has been taken by \citet{MenonL16}, who consider profiles of top orders in which most voters have single-peaked preferences.
Furthermore, work on matrices that almost satisfy the consecutive ones property (see~Section~\ref{sec:consones}) is highly relevant for this purpose; we refer the reader to a survey by \citet{dom2009consecutive} for a literature overview. It would be of great interest to see how close to being possibly single-peaked the profiles on PrefLib are.

Finally, work by \citet{durand2003finding}, \citet{LacknerL-likelihoodSP}, and \citet{chen2018number} analyzed the number of single-peaked profiles of total orders. To extend this work to profiles of partial orders could shed light on how much generality can be gained by moving from total orders to partial orders.
Even more interesting would be a comparison of single-peaked concepts for weak orders (see~Section~\ref{sec:other}). Clearly, the possibly single-peaked definition is more general than, e.g., the single-plateaued definition---but a quantitative answer is missing.

\section*{Acknowledgements}

This work is based on a conference paper by \citet{incompletesp} and a follow-up conference paper by \citet{fitzsimmons2014single}.

Martin Lackner was supported by the Austrian Science Foundation FWF, grant P31890.
Zack Fitzsimmons was supported in part by the National Science Foundation under grant
no.\ CCF-1101452 and by a National Science Foundation Graduate Research
Fellowship under grant no.\ DGE-1102937. The research was done in part while Zack
Fitzsimmons was on research leave at Rensselaer Polytechnic Institute.
The collaboration between the two authors of this paper has been supported by COST Action IC1205 on Computational Social Choice.

We thank Edith Elkind, Piotr Faliszewski, Edith Hemaspaandra, David Narv{\'a}ez, Stanis{\l}aw Radziszowski, and the anonymous referees for their
helpful comments and suggestions.
We thank Dominik Peters for helpful discussions and the permission to use the Ti\textit{k}Z script that produced Figure~\ref{fig:weak-orders-overview}.

\newpage
\appendix
\section{Proof details}

\subsection{Correctness of the Guided Algorithm}
\label{sec:app:corr-guided}

We are going to prove that the Guided Algorithm (Algorithm~\ref{alg:untie}) is correct.
In the following, we write $a\succsim b$ to denote that either $a \succ b$ or $a\sim b$.
We consider the Guided Algorithm at any given point during its runtime.
In particular, we consider the sets $A_L$ and $A_R$ as constructed by the algorithm.
Let $V_g$ be the guiding vote in $\mathcal{P}$.

\begin{lemma}
Let $k\in[n]$, $a_{\ell}=\max_k(A_L)$, $a_r=\max_k(A_R)$, and $c_j \in \{c_i,\ldots, c_m\}$.
It either holds that $c_j\succsim_k a_r$ or that $c_j\succsim_k a_{\ell}$.\label{lem:alllarger}
\end{lemma}

\begin{proof}
Without loss of generality we assume that $a_{\ell}$ was placed before $a_r$.
Towards a contradiction assume that both $a_r \succ_k c_j$ and $a_{\ell} \succ_k c_j$.
Let us consider the algorithm at the point when $a_r$ was placed ($c_{i'}=a_r$ for some $i'<i$).
We will show that (R1) is true and thus $a_r$ could not have been placed on the right-hand side.
Recall condition (R1):
\begin{align*}
c_i \succ_k \min_k(C_{>i}) &\text{\quad and \quad}
\max_k(A_L) \succ_k \min_k(C_{>i})
\end{align*}
Since $c_{i'} = a_r \succsim_k c_j \succsim_k \min_k(C_{>i'})$ and $\max_k(A_L)= a_{\ell}  \succ_k c_j \succsim_k \min_k(C_{>i'})$, (R1) is true and $c_{i'}$ could not have been placed on the right side of the axis.
\end{proof}

\begin{lemma}
Let $k\in[n]$ and $c_j = \max_k(c_i,\ldots, c_m)$.
Furthermore, let $a,a'\in A_R$ such that candidate $a$ has been placed on $A_R$ before $a'$.
Then it either holds that $a' \succsim_k c_j$ or it holds that $a' \succsim_k a$.
\label{lem:twoonsameside}%
\end{lemma}
\begin{proof}
We consider the algorithm at the point where $a'$ was placed on the right-hand side, i.e., in $A_R$.
At this point, condition (R2) had to be false; the fact that $a' \succsim_k c_j$ or $a' \succsim_k a$ holds is a direct consequence.
\end{proof}

\begin{proposition}
The Guided Algorithm (Algorithm~\ref{alg:untie}) is correct, i.e., it outputs an axis if and only if the given preference profile is single-peaked, and, furthermore, $\mathcal{P}$ is single-peaked with respect to any axis that is returned by the algorithm.
\end{proposition}

\begin{proof}
We first show that if an axis $A$ is found, the profile $\mathcal{P}$ is single-peaked with respect to $A$.
Note that the guiding vote $V_g$ is single-peaked with respect to $A$, as the Guided Algorithm constructs $A$ based on $V_g$.
Towards a contradiction assume that there is a vote $V\in \mathcal{P}$ that is not single-peaked with respect to $A$.
This means that there are three candidates $a,b,c$ with order $a\rhd b\rhd c$ on $A$, $a\succ b$ and $c\succ b$.
We have to distinguish six cases of how $a,b,c$ are ordered by the guiding vote $\succ_g$:
\begin{itemize}
\item $a\prec_g b\prec_g c$ ($a$ is placed first, then $b$, then $c$; other candidates in arbitrary order):
Let us consider the algorithm at the point when $b$ is being placed, i.e., $b=c_i$, and when the conditions for vote $V$ are being checked.
It holds that either $a\in A_L$ or $a\in A_R$.
If $a\in A_L$, observe that condition~\eqref{eq:cond4} is satisfied since $a\succ b$ and $c\succ b$.
Consequently, $b$ has to be placed on the right side ($\textit{left} =\false$).
Then it holds that $a\rhd c\rhd b$ on the axis generated by the algorithm which contradicts our assumption that $a\rhd b\rhd c$ holds.
In the case that $a\in A_R$ condition~\eqref{eq:cond2} is satisfied.
This leads to a contradiction by the same argument.
\item $c\prec_g b\prec_g a$:
This case is analogous.
\item $a\prec_g c\prec_g b$:
Now we consider the point where $c$ is being placed, i.e., $c=c_i$, and when the conditions for vote $V$ are being checked.
It holds that either $a\in A_L$ or $a\in A_R$.
If $a\in A_L$, observe that \eqref{eq:cond1} is satisfied and hence $c$ has to be placed on the left side ($\textit{right} =\false$).
Then it holds that $a\rhd c\rhd b$ on the axis generated by the algorithm which contradicts our assumption that $a\rhd b\rhd c$ holds.
In the case that $a\in A_R$, condition~\eqref{eq:cond3} is satisfied and we obtain a contradiction by the same argument.
\item $c\prec_g a\prec_g b$:
This case is analogous to the previous one.
\item $b\prec_g c\prec_g a$ or $b\prec_g a\prec_g c$:
Since $b$ is placed on $A$ before $a$ and $c$, the resulting axis cannot be $a\rhd b\rhd c$, a contradiction.
\end{itemize}

For the other direction, let us show that if the algorithm returns \nsp, then the profile $\mathcal{P}$ is not single-peaked.
First, let us observe under what conditions the algorithm returns \nsp.
There are four cases: Either conditions \eqref{eq:cond1} and \eqref{eq:cond3}, \eqref{eq:cond1} and \eqref{eq:cond4}, \eqref{eq:cond2} and \eqref{eq:cond3} or conditions \eqref{eq:cond2} and \eqref{eq:cond4} hold.
These pairs of conditions may either hold for the same vote or for two distinct votes; we denote these two votes $V$ and $V'$ although it might be that these two are the same.
\begin{itemize}
\item While placing $c_i$, condition~\eqref{eq:cond3} holds for some vote $V$ and condition~\eqref{eq:cond1} holds for some vote $V'$:\\
We have the following five candidates in these conditions: $a_r=\max_{V}(A_R)$, $c_i$, $c_j=$ $\min_{V}(C_{>i})$ in condition~\eqref{eq:cond3} and $a_{\ell}=\max_{V'}(A_L)$, $c_i$, $c'_j=\min_{V'}(C_{>i})$ in condition~\eqref{eq:cond1}.
In Figure~\ref{fig:proof-condL1andR1}, the known information about the votes $V$ and $V'$ is shown.
Since condition~\eqref{eq:cond1} and \eqref{eq:cond3} are symmetrical, we can assume without loss of generality that $a_{\ell}$ is placed before $a_r$.
Thus, we can assume that the guiding vote is $\{ c_j, c'_j\} \succ_g c_i\succ_g a_r\succ_g a_{\ell}$, as shown in the figure (the order of $c_j$ and $c_j'$ is not relevant and thus can be arbitrary).

There are four types of axes possible that are compatible with the guiding vote $V_g$. 
(The order of candidates in sets is arbitrary.)
\begin{itemize}
\item $\left\langle a_{\ell} \rhd  c_i \rhd  \{c_j, c_j'\} \rhd  a_r  \right\rangle$:
Vote $V$ is not single-peaked with respect to any axis of this type (or their reverse).
\item $\left\langle a_{\ell} \rhd  \{c_j, c_j'\} \rhd  c_i \rhd  a_r  \right\rangle$:
Vote $V'$ is not single-peaked with respect to any axis of this type (or their reverse).
\item $\left\langle a_{\ell} \rhd  a_r \rhd  \{c_j, c_j'\} \rhd  c_i \right\rangle$:
Vote $V$ is not single-peaked with respect to any axis of this type (or their reverse).
\item $\left\langle a_{\ell} \rhd  a_r \rhd  c_i \rhd  \{c_j, c_j'\} \right\rangle$:
Consider vote $V'$ and Lemma~\ref{lem:alllarger}.
Since $a_{\ell} \succ' c'_j$ it has to hold that $c'_j \succsim' a_r$.
Consequently, $a_{\ell} \succ' a_r$ and $c_i \succ' a_r$.
Thus, the candidates $a_{\ell}, a_r, c_i$ form a v-valley for vote $V'$.
\end{itemize}
Since these are all possible axes, we can conclude that the profile is not single-peaked.

\item While placing $c_i$, condition~\eqref{eq:cond4} holds for some vote $V$ and condition~\eqref{eq:cond2} holds for some vote $V'$:\\
This case is similar to the previous one.
In particular, we use the same candidate variables.
In Figure~\ref{fig:proof-condL2andR2}, the known information about the votes $V$ and $V'$ is shown.

\begin{figure}
    \centering
    \begin{minipage}{0.48\textwidth}
		\centering 
		{\renewcommand{\tabcolsep}{0.9em}
		\begin{tabular}{cccc}
			Guiding vote $V_g$ & Vote $V$ & Vote $V'$ \\
			\begin{tikzpicture}
			\tikzstyle{every node}=[circle, inner sep= 0.07cm];
			\tikzstyle{every path}=[thick]		;	
			\node (5) at (-0.4,3*0.9) {$c_j'$};
			\node (4) at (0.4,3*0.9) {$c_j$};
			\node (3) at (0,2*0.9) {$c_i$};
			\node (2) at (0,1*0.9) {$a_r$};
            \node (1) at (0,0*0.9) {$a_{\ell}$};
			\draw (4)  -- (3) -- (2) -- (1);
			\draw (5) --  (3);
			\end{tikzpicture}
			& 
			\begin{tikzpicture}
			\tikzstyle{every node}=[circle, inner sep= 0.07cm];
			\tikzstyle{every path}=[thick];
			\node (3a) at (-0.4,2*1.2) {$a_r$};
			\node (3b) at (0.4,2*1.2) {$c_i$};
			\node (2) at (0,1*1.2) {$c_j$};
			\node (invisible) at (0,0) {};
			\draw (3a) -- (2);
			\draw (3b) --(2);
			\end{tikzpicture}
			& 
			\begin{tikzpicture}
			\tikzstyle{every node}=[circle, inner sep= 0.07cm];
			\tikzstyle{every path}=[thick];
            \node (3a) at (-0.4,2*1.2) {$a_{\ell}$};
			\node (3b) at (0.4,2*1.2) {$c_i$};
			\node (2) at (0,1*1.2) {$c_j'$};
			\node (invisible) at (0,0) {};
			\draw (3a) -- (2);
			\draw (3b) --(2);
			\end{tikzpicture}
		\end{tabular}
		}
    	\caption{Conditions (L1) and (R1)}\label{fig:proof-condL1andR1}
    \end{minipage}\hfill
    \begin{minipage}{0.48\textwidth}
        \centering
		{\renewcommand{\tabcolsep}{0.9em}
		\begin{tabular}{cccc}
			Guiding vote $V_g$ & Vote $V$ & Vote $V'$ \\
			\begin{tikzpicture}
			\tikzstyle{every node}=[circle, inner sep= 0.07cm];
			\tikzstyle{every path}=[thick];
			\node (5) at (-0.4,3*0.9) {$c_j'$};
			\node (4) at (0.4,3*0.9) {$c_j$};
			\node (3) at (0,2*0.9) {$c_i$};
			\node (2) at (0,1*0.9) {$a_r$};
            \node (1) at (0,0*0.9) {$a_{\ell}$};
			\draw (4)  -- (3) -- (2) -- (1);
			\draw (5) --  (3);
			\end{tikzpicture}
			& 
			\begin{tikzpicture}
			\tikzstyle{every node}=[circle, inner sep= 0.07cm];
			\tikzstyle{every path}=[thick];
            \node (3a) at (-0.4,2*1.2) {$a_{\ell}$};
			\node (3b) at (0.4,2*1.2) {$c_j$};
			\node (2) at (0,1*1.2) {$c_i$};
			\node (invisible) at (0,0) {};
			\draw (3a) -- (2);
			\draw (3b) --(2);
			\end{tikzpicture}
			& 
			\begin{tikzpicture}
			\tikzstyle{every node}=[circle, inner sep= 0.07cm];
			\tikzstyle{every path}=[thick];
			\node (3a) at (-0.4,2*1.2) {$a_r$};
			\node (3b) at (0.4,2*1.2) {$c_j'$};
			\node (2) at (0,1*1.2) {$c_i$};
			\node (invisible) at (0,0) {};
			\draw (3a) -- (2);
			\draw (3b) --(2);
			\end{tikzpicture}
		\end{tabular}
		}
		\caption{Conditions (L2) and (R2)}\label{fig:proof-condL2andR2}
    \end{minipage}
\end{figure}

There are four types of axes possible that are compatible with $V_g$. 
(The order of candidates in sets is arbitrary.)
\begin{itemize}
\item $\left\langle a_{\ell} \rhd  c_i \rhd  \{c_j, c_j'\} \rhd  a_r  \right\rangle$:
Vote $V$ is not single-peaked with respect to any axis of this type (or their reverse).
\item $\left\langle a_{\ell} \rhd  \{c_j, c_j'\} \rhd  c_i \rhd  a_r  \right\rangle$:
Vote $V'$ is not single-peaked with respect to any axis of this type (or their reverse).
\item $\left\langle a_{\ell} \rhd  a_r \rhd  \{c_j, c_j'\} \rhd  c_i \right\rangle$:
Consider vote $V$ and Lemma~\ref{lem:alllarger}.
Since $a_{\ell} \succ c_i$ it has to hold that $c_i \succsim a_r$.
Thus, the candidates $a_{\ell}, a_r, c_j$ form a v-valley for vote $V$.
\item $\left\langle a_{\ell} \rhd  a_r \rhd  c_i \rhd  \{c_j, c_j'\} \right\rangle$:
Vote $V'$ is not single-peaked with respect to any axis of this type (or their reverse).
\end{itemize}

\item While placing $c_i$, condition~\eqref{eq:cond3} holds for vote $V$ and condition~\eqref{eq:cond2} holds for vote $V'$:\\
We have the following five candidates in these conditions: $a_r=\max_{V}(A_R)$ $c_i$, $c_j=\min_{V}(C_{>i})$ in condition~\eqref{eq:cond3} and $a'_r=\max_{V'}(A_R)$, $c_i$, $c'_j=\min_{V'}(C_{>i})$ in condition~\eqref{eq:cond2}.
In Figure~\ref{fig:proof-condL1andR2}, the known information about the votes $V$ and $V'$ is shown.
In the following arguments it is irrelevant which of $c_j$ and $c_j'$ is placed first.
However, for $a_r$ and $a_r'$ this is relevant.
We will consider both possibilities (cases 1 and 2 in Figure~\ref{fig:proof-condL1andR2}).
There are four types of axes possible that are compatible with either of these two guiding votes. 
(The order of candidates in sets is arbitrary.)
\begin{itemize}
\item $\left\langle \{a_r,a_r'\} \rhd  \{c_j, c_j'\} \rhd  c_i  \right\rangle$:
Vote $V$ is not single-peaked with respect to any axis of this type (or their reverse).
\item $\left\langle \{a_r,a_r'\} \rhd  c_i \rhd  \{c_j, c_j'\}  \right\rangle$:
Vote $V'$ is not single-peaked with respect to any axis of this type (or their reverse).
\item $\left\langle a_r \rhd  \{c_j, c_j'\} \rhd  c_i \rhd  a_r' \right\rangle$:
Both $V$ and $V'$ are not single-peaked with respect to any axis of this type (or their reverse).
\item $\left\langle a_r \rhd  c_i \rhd  \{c_j, c_j'\} \rhd  a_r' \right\rangle$:
Here we have to distinguish whether $a_r$ or $a_r'$ is placed first (case 1 and 2 in Figure~\ref{fig:proof-condL1andR2}).
\begin{figure}
\begin{center}
{\renewcommand{\tabcolsep}{0.9em}
\begin{tabular}{cccc}
	Guiding vote $V_g$ & Guiding vote $V_g$ & Vote $V$ & Vote $V'$ \\
	(case 1) & (case 2) & &\\
	\begin{tikzpicture}
	\tikzstyle{every node}=[circle, inner sep= 0.07cm];
	\tikzstyle{every path}=[thick]
	\node (5) at (-0.4,3) {$c_j'$};
	\node (4) at (0.4,3) {$c_j$};
	\node (3) at (0,2) {$c_i$};
	\node (2) at (0,1) {$a_r$};
	\node (1) at (0,0) {$a_r'$};
	\draw (4)  -- (3) -- (2) -- (1);
	\draw (5) --  (3);
	\end{tikzpicture}
	& 
	\begin{tikzpicture}
	\tikzstyle{every node}=[circle, inner sep= 0.07cm];
	\tikzstyle{every path}=[thick]
	\node (5) at (-0.4,3) {$c_j'$};
	\node (4) at (0.4,3) {$c_j$};
	\node (3) at (0,2) {$c_i$};
	\node (2) at (0,1) {$a_r'$};
	\node (1) at (0,0) {$a_r$};
	\draw (4)  -- (3) -- (2) -- (1);
	\draw (5) --  (3);
	\end{tikzpicture}
	& 
	\begin{tikzpicture}
	\tikzstyle{every node}=[circle, inner sep= 0.07cm];
	\tikzstyle{every path}=[thick]
	\node (3a) at (-0.4,2*1.2) {$c_i$};
	\node (3b) at (0.4,2*1.2) {$a_r$};
	\node (2) at (0,1*1.2) {$c_j$};
	\node (invisible) at (0,0) {};
	\draw (3a) -- (2);
	\draw (3b) --(2);
	\end{tikzpicture}
	& 
	\begin{tikzpicture}
	\tikzstyle{every node}=[circle, inner sep= 0.07cm];
	\tikzstyle{every path}=[thick]
	\node (3a) at (-0.4,2*1.2) {$c_j'$};
	\node (3b) at (0.4,2*1.2) {$a_r'$};
	\node (2) at (0,1*1.2) {$c_i$};
	\node (invisible) at (0,0) {};
	\draw (3a) -- (2);
	\draw (3b) --(2);
	\end{tikzpicture}
\end{tabular}
}
\end{center}
\caption{Condition (L1) and Condition (R2)}
\label{fig:proof-condL1andR2}
\end{figure}
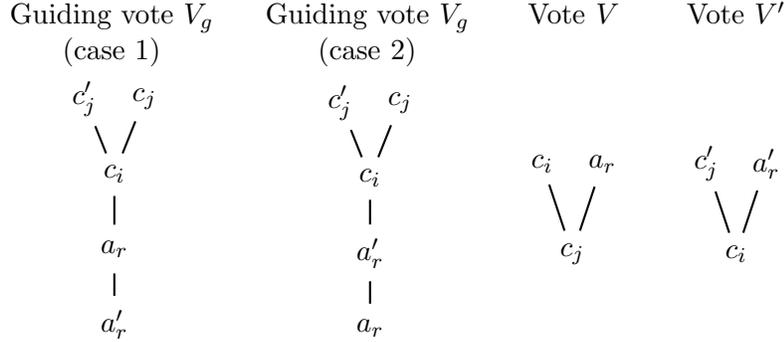
\begin{itemize}
\item[Case 1 ($a_r'\prec_g a_r$):]
Lemma~\ref{lem:twoonsameside} applied to $V'$ implies 
that either (i) $a_r \succsim' c_i$ or (ii) $a_r \succsim' a_r'$.
In both cases we encounter a valley with respect to vote $V'$ for the candidates $\{a_r,a_r',c_i\}$.
\item[Case 2 ($a_r\prec_g a_r'$):]
This case is similar to the previous.
Lemma~\ref{lem:twoonsameside} applied to $V$ implies 
that either (i) $a_r \succsim c_j$ or (ii) $a_r' \succsim a_r$.
In both cases we encounter a valley with respect to vote $V$ for the candidates $\{a_r',a_r,c_j\}$.
\end{itemize}
We see that in both cases either $V$ or $V'$ is not single-peaked with respect to any axis of this type (or their reverse).
\end{itemize}

\item While placing $c_i$, condition~\eqref{eq:cond4} holds for some vote and condition~\eqref{eq:cond1} holds for some vote:\\
This can be shown analogously to the previous case since \eqref{eq:cond1} and \eqref{eq:cond3} are symmetrical as well as \eqref{eq:cond2} and \eqref{eq:cond4}, see~Figure~\ref{fig:cond}.
\end{itemize}
We have shown that if the algorithm returns \nsp then the profile $\mathcal{P}$ is indeed not single-peaked.
\end{proof}

\subsection{Runtime and Correctness for the Unguided Algorithm}
\label{sec:app:unguided}

\begin{reptheorem}{thm:unguided}\thmunguidedtext
\end{reptheorem}

\begin{proof}
Let us first discuss the runtime of Algorithm~\ref{alg:unguided}. 
Some functions have to be precomputed to achieve the $\bigO(m^2\cdot n)$ runtime,
 in particular the function \texttt{IntersectingVote}.
The function $\texttt{IntersectingVote}(A)$ returns a vote $V$ with $A\cap V\neq \emptyset$ and $V\setminus A\neq \emptyset$.
We show that it suffices to compute a list of $2m$ votes to answer \texttt{IntersectingVote} function calls in constant time.
Let us first make the following observation:
Let $c\in C$.
Consider the set of votes for which the sets $\{c'\in C \mid c' \succ c\}$ are maximal (with respect to $\subseteq$).
If we consider a single-peaked axis, then candidates in such a set have to form a contiguous subsequence either directly left or directly right of $c$.
Since these sets are maximal, only two of them can exist (assuming single-peakedness).
Consequently, we compute these maximal sets for each candidate.
If three or more exist for one candidate, we can terminate the algorithm already at this point.
Also, if two maximal sets have a non-empty intersection, the algorithm terminates.
(The candidates in the intersection would have to lie left and right of $c$).
Otherwise we store the (at most) two corresponding votes for each candidate.

Let $A=\left\langle a_1 \rhd  \ldots \rhd  a_{i-1} \rhd  a_{i} \right\rangle$, i.e., $a_i$ is the rightmost candidate in the incomplete axis $A$.
The function call $\texttt{IntersectingVote}(A)$ can now be answered by considering the one or two maximal votes for $a_i$.
The function simply returns the vote where $a_{i-1}$ is not ranked higher than $a_i$.
(It might be that both votes do not rank $a_{i-1}$ higher than $a_i$.
In this case $A$ cannot be extended to a single-peaked axis, but this is going to be detected by algorithm.
Any of the two axes can be returned.)

It remains to observe that finding the (at most two) maximal votes for a candidate $c$ requires $\bigO(m\cdot n)$ time.
This has to be done for every candidate and consequently this preprocessing requires $\bigO(m^2\cdot n)$ time.

We can now analyze the runtime of the algorithm.
The main loop (line~\ref{line:cstart}) iterates over all $m$ candidates.
The loop in line~\ref{line:foreachvotewithpeaki} iterates over every vote at most once.
Consequently, the $\oplus$ operator is applied at most $n$ times.
Since $A\oplus V$ can be computed in $\bigO(m)$ time, the lines~\ref{line:foreachvotewithpeaki} to \ref{line:end:foreachvotewithpeaki} have a total runtime of $\bigO(m^2\cdot n)$.

It remains to determine the runtime of the lines~\ref{line:nomorepeakinA} to \ref{line:end:nomorepeakinA}.
Due to the preprocessing of the \texttt{IntersectingVote} procedure we can obtain $V$ in constant time.
The profile $\mathcal{P}'$ can be generated in $\bigO(\card{C'}\cdot n)$ time.
Applying the Guided Algorithm requires $\bigO(\card{C'}\cdot n)$ time as well (Theorem~\ref{thm:guided-weak}).
Observe that after applying the Guided Algorithm the candidates in $C'$ are placed on the axis.
Consequently, the Guided Algorithm is always applied to a disjoint set of candidates (except maybe $a_i$).
Hence for a fixed $c_\text{start}\in C$, the total runtime of the Guided Algorithm is $\bigO(m\cdot n)$.
Taking the loop in line~\ref{line:cstart} into account, we obtain a total runtime of $\bigO(m^2\cdot n)$.

Let us now show that the Unguided Algorithm (Algorithm~\ref{alg:unguided}) is correct, i.e., it outputs an axis if and only if the given preference profile is single-peaked and, furthermore, $\mathcal{P}$ is single-peaked with respect to any axis that is returned by the algorithm.
If the algorithm outputs an axis, it is certainly single-peaked since this is tested for every vote in line~\ref{line:oplus}.
By the same argument, one can conclude that if the profile is not single-peaked, the algorithm returns \nsp.

It remains to prove that the algorithm always returns a valid axis in case of a single-peaked profile.
Let us consider the algorithm at the time when $c_\text{start}$ is the leftmost candidate of a valid axis.
We show that the algorithm will find a complete axis with $c_\text{start}$ as leftmost candidate.
First, observe that for every $i$ in line~\ref{line:fori} either there is a vote $V$ with $a_i$ as its top-ranked candidate or the condition in line~\ref{line:nomorepeakinA} is true.
In the first case, it is easy to verify that the $\oplus$ operator adds candidates to the axis in the only possible way.
In the second case, 
an intersecting vote is found as guaranteed by the connectedness condition.
Then the Guided Algorithm is applied and the axis is extended by the candidates in $C'$.
It remains to verify that the resulting axis $A$ is single-peaked for all votes with a non-empty intersection with $C'$.
This is guaranteed by the $x$ element, which ensures that candidates outside of $A\cup C'$ are taken into account.
\end{proof}

\bibliography{litfull}
\bibliographystyle{theapa}

\end{document}